\documentclass[10pt,twocolumn]{IEEEtran}
\usepackage{enumerate}
\usepackage{mathrsfs}
\usepackage{amsmath}
\usepackage{amssymb}
\usepackage{amsthm}
\usepackage[labelformat=empty,labelsep=none]{caption}
\usepackage{color}

\ifCLASSINFOpdf \else \fi
\ifCLASSINFOpdf \else \fi
\usepackage[pdftex]{graphicx}
\usepackage{algorithm}
\usepackage{algorithmic}
\usepackage{subfigure}
\usepackage{multirow}
\usepackage{multirow}
\usepackage{placeins}
\usepackage{afterpage}
\usepackage{stfloats}
\usepackage{caption}
\usepackage{makecell}

\newtheorem{lemma}{Lemma}

\usepackage{color}
\usepackage{slashbox}
\usepackage[normalem]{ulem}

\definecolor{c}{RGB}{0,112,192}

\begin{document}

\title{Neighbor Discovery for VANET with Gossip Mechanism and Multi-packet Reception}

\author{Zhiqing Wei,~\IEEEmembership{Member,~IEEE,}
	Qian Chen,
	Heng Yang,~\IEEEmembership{Student Member,~IEEE,}\\
	Huici Wu,~\IEEEmembership{Member,~IEEE,}
	Zhiyong Feng,~\IEEEmembership{Senior Member,~IEEE,}
	Fan Ning
	
\thanks{This work is supported by the Beijing Natural Science Foundation (No. L192031) and National Natural Science Foundation of China (No. 61631003).
		
The authors are with Beijing University of Posts and Telecommunications, Beijing, China 100876 (email: \{weizhiqing, fengzy, dailywu\}@bupt.edu.cn).
\emph{Correspondence authors: Zhiqing Wei and Zhiyong Feng.}}}

\maketitle

\begin{abstract}
	Neighbor discovery (ND) is a key initial step of network configuration and prerequisite of vehicular ad hoc network (VANET).
	However, the convergence efficiency of ND is facing the requirements of multi-vehicle fast networking of VANET with frequent topology changes.
	This paper proposes gossip-based information dissemination and sensing information assisted ND with MPR (GSIM-ND) algorithm for VANET.
	GSIM-ND algorithm leverages efficient gossip-based information dissemination in the case of multi-packet reception (MPR).
	Besides, through multi-target detection function of multiple sensors installed in roadside unit (RSU),
	RSU can sense the distribution of vehicles and help vehicles to obtain the distribution of their neighbors.
	Thus, GSIM-ND algorithm leverages the dissemination of sensing information as well.	
	The expected number of discovered neighbors within a given period is theoretically derived and used as the critical metric to evaluate the performance of GSIM-ND algorithm.
	The expected bounds of the number of time slots when a given number of neighbors needs to be discovered is derived as well.
	The simulation results verify the correctness of theoretical derivation.
	It is discovered that GSIM-ND algorithm proposed in this paper can always reach the short-term convergence quickly.
	Moreover, GSIM-ND algorithm is more efficient and stable compared with completely random algorithm (CRA), scan-based algorithm (SBA) and gossip-based algorithm.
	The convergence time of GSIM-ND algorithm is 40\%-90\% lower than that of these existing algorithms for both low density and high density networks.
	Thus, GSIM-ND can improve the efficiency of ND algorithm.
\end{abstract}


\section{Introduction}\label{sec_1}
	In vehicular ad hoc network (VANET), the vehicle equipped with on-board unit (OBU) is a mobile road node and belongs to the road user, roadside unit (RSU) is a fixed road node and belongs to the road server.
	OBU is a microwave device used to communicate with RSU.
	RSU plays an important role in the distribution and dissemination of traffic data \cite{1}, \cite{2}, i.e., RSU is a key component of collaborative and distributed applications in VANET \cite{3} and can provide services for road users such as OBUs \cite{4}.
	The applications in VANET, especially road safety application, mainly depend on the communication between OBU and OBU (V2V), OBU and RSU (V2I) \cite{5}-\cite{7}, collectively known as V2X communication.
	
	As a prerequisite for VANET, fast neighbor discovery (ND) is a key initial step of network configuration.
	The convergence efficiency of ND directly affects the performance of the network \cite{9}.
	The nodes that can directly communicate without relay nodes are one-hop neighbors of each other.
	For VANET with frequent changing topology, all the nodes need to discover their one-hop neighbors quickly, establish an effective communication infrastructure connection, and adapt to frequent changing topology \cite{10}.
	Then, the information can be transmitted directly or indirectly through the connections established in the process of ND when a node needs to report emergency or traffic congestion information to other nodes.
	
	In the era of fifth-generation (5G) mobile communication and beyond, with the requirements of high-capacity services in VANETs,
	the communication frequency band will gradually extend to the millimeter wave band.
	Therefore, the directional antenna for the VANETs on millimeter-wave band is being considered to be used in the standardizations such as \cite{11}, \cite{12}.
	The directional antenna can concentrate the energy in a certain beam direction, which results in higher capacity, stronger spatial diversity and larger transmission range.
	In this paper, we consider ND based on directional antenna.
	
	However, the ND with directional antenna faces the challenges of beam alignment and fast convergence.
	To address these challenges, gossip-based information dissemination can be used to assist ND in VANET.
	Since V2X communication uses orthogonal frequency division multiple access (OFDMA) technology \cite{13},
	communication receiver has strong information reception capability.
	Therefore, ND in VANET can be conducted with multi-packet reception (MPR).
	In addition, to guarantee the correctness and efficiency of information dissemination in VANET, some literatures investigated VANET with integrated sensing and communication RSU \cite{14}-\cite{18}.
	In this scenario, both communication and sensing functions are implemented on RSUs.
	Then, i) RSUs with sensing and communication functions not only passively wait for vehicles to upload road traffic data \cite{1}, \cite{2}, but also actively detect road conditions.
	ii) The radar function of RSU can play the role of predictive beamforming design and re-authentication for communication data \cite{14}, \cite{15}.
	iii) RSU can detect vehicles and forward the information to the vehicles such that the problem of information asymmetry caused by equipment inconsistency between different types of vehicles can be avoided \cite{18}.
	Since the prior information of the neighbors can accelerate the convergence of ND \cite{19},
	VANET with sensing function \cite{14}-\cite{18} can take sensing information as the prior information to accelerate ND.
	
	This paper studies gossip-based information dissemination and sensing information assisted ND with MPR (GSIM-ND) in VANET.
	Since there may be blind spots in the sensing range of RSU,
	different convergence conditions of GSIM-ND algorithm for mobile nodes at different locations are given as well.
	It is noted that GSIM-ND algorithm still works well when there is no sensing information or the accuracy of sensing information is not high.
	For the proposed GSIM-ND algorithm, the expected number of discovered neighbors within a given period is derived and adopted as the critical metric to evaluate the performance of GSIM-ND algorithm.
	The simulation results verify the correctness of theoretical derivation.
	It is discovered that GSIM-ND algorithm can always reach the short-term convergence and converge efficiently.
	GSIM-ND algorithm proposed in this paper can adapt to the requirements of multi-vehicle fast networking better than the traditional ND algorithms.
	
	The remainder of this paper is organized as follows.
	Section \ref{sec_2} reviews the related works.
	Section \ref{sec_3} introduces the system model.
	Section \ref{sec_4} gives a detailed design of GSIM-ND algorithm.
	Section \ref{sec_5} derives the expected number of discovered neighbors with GSIM-ND algorithm.
	Section \ref{sec_6} provides simulation results and analysis for the proposed GSIM-ND algorithm.
	Section \ref{sec_7} summarizes this paper.
	The main parameters in this paper are summarized in Table \ref{table1}.
	\begin{table}[h]
		\centering
		\renewcommand\arraystretch{1.5}
		\caption{Table I \\Key notations in this paper}\label{table1}
		\begin{tabular}{p{1.4cm}<{\centering}p{6.6cm}}
			\hline
			Notation&Description\\
			\hline
			$d$&Road width\\
			$L$&Road length\\
			$r$&Communication and sensing radius of nodes\\	
			$\theta$&Beamwidth of directional antenna\\
			$\alpha$&Probability of a mobile node selecting one of non-empty beam directions\\
			$M$&Number of mobile nodes in the network\\
			$\rho$&Density of mobile nodes\\
			$N_{\rm b}$&Number of neighbors in a non-empty beam\\			
			$N_{\rm I}$&Number of common neighbors between two mobile nodes\\
			$N$&Number of a mobile node's neighbors\\
			$k$&Number of different modulation modes\\
			$\overline {t_{\rm u}(n)}$&Expected upper bound of the number of time slots when $n$ neighbors need to be directly or indirectly discovered\\
			$\overline {t_{\rm l}(n)}$&Expected lower bound of the number of time slots when $n$ neighbors need to be directly or indirectly discovered\\
			$\overline {n(t)}$&Expected number of directly or indirectly discovered neighbors within $t$ time slots\\
			\hline
		\end{tabular}
	\end{table}

\section{Related Works}\label{sec_2}
\subsection{ND Algorithm}
	The essence of ND is to directly or indirectly transmit and receive the packets carrying neighbors' identity information among neighbors by \textit{handshake}, establish connections among neighbors, and maintain neighbor information table.
	The optimization goal of ND is to find complete and correct neighbors as soon as possible.

	ND algorithm can be classified into two categories when the directional antenna is adopted.
	One is completely random algorithm (CRA).
	The nodes randomly select transmit mode or receive mode and select a beam direction with a certain probability in each time slot.
	The other is scan-based algorithm (SBA).
	In each time slot, the nodes point the antenna at the corresponding beam direction in turn according to the preset scan sequence and select transmit mode or receive mode randomly or determinately \cite{20}-\cite{24}.
	CRA has a high probability of completing ND within a certain time.
	SBA can complete ND within a certain time deterministically \cite{25}.
	However, SBA requires nodes not to move during a scan period whereas CRA only requires nodes not to move within a time slot which is much shorter than a scan period.
	Therefore, CRA is more suitable for scenarios with high mobility \cite{20}, such as the scenario of VANET.
\subsection{Improvement of ND Algorithm}
	In order to improve the convergence speed of ND,
	existing literatures have carried out research on the prior information and the dissemination mode of information respectively.
	
	In terms of the prior information,
	Burghal \textit{et al.} in \cite{19} discussed the impact of prior information of the set of neighbors on ND.
	The distribution and number of neighbors detected by radar can accelerate ND \cite{8}, \cite{26}, \cite{27}.
	Vasudevan \textit{et al.} in \cite{25} and Sun \textit{et al.} in \cite{28} both found that the unknown number of neighbors would lead to at most a two-factor deceleration in the rate of ND convergence.
	According to the collision information provided by a collision feedback mechanism, Khalili \textit{et al.} in \cite{29} estimated the number of the node's neighbors such that they adjusted the transmission probability, which accelerates the convergence of ND.
	Vasudevan \textit{et al.} in \cite{30} estimated the number of neighbors when there was no prior information of the number of neighbors.
	They found that ND was faster as the estimated value of the number of neighbors accurate.
	Therefore, VANET with sensing function \cite{14}-\cite{18} can take the sensing information as the prior information to accelerate the convergence of ND.
	
	In terms of the dissemination mode of information,
	the ND algorithm mentioned above can only discover neighbors through direct \textit{handshake} with one neighbor at a time, which belongs to the direct ND algorithm with SPR.
	Vasudevan \textit{et al.} in \cite{30} proposed a gossip-based ND algorithm that can realize ND through direct or indirect \textit{handshake}.
	Compared with the direct ND algorithm, the gossip-based ND algorithm could greatly accelerate the generation of new \textit{handshake} through the \textit{handshake} that has been completed.
	Russell \textit{et al.} in \cite{31} analyzed ND with MPR.
	Compared with SPR, MPR enables nodes to \textit{handshake} with multiple neighbors at a time, reducing collisions and enabling more efficient ND.
	Vasudevan \textit{et al.} in \cite{30} and \cite{31} studied ND algorithm in the scenario of wireless ad hoc networks.
	Considering that the environmental characteristics of VANET are different from the other wireless networks, this paper studies GSIM-ND algorithm for VANET.
\subsection{ND Algorithm for VANET}
	For VANET, existing literatures have carried out ND research to improve the probability of successful neighbor \textit{handshake} by means of trajectory prediction, vehicle clustering, packet loss estimation, etc.
	
	In terms of trajectory prediction,
	Liu \textit{et al.} in \cite{32} predicted the trajectories of vehicle mobile nodes based on Kalman filter and studied the method for vehicle mobile nodes to periodically update the neighbors.
	However, a simply broadcast mechanism was adopted to disseminate new neighbor information in the process of neighbor update.
	Too many broadcast packets cause high probability of collisions, which is not beneficial to the fast convergence of neighbor update.
	In terms of vehicle clustering,
	Nahar \textit{et al.} in \cite{33} limited the dissemination of neighbor information within a cluster of vehicles and probabilistic forwarded or discarded packets,
	which reduced unnecessary information broadcasting and avoided storm information.
	In terms of packet loss estimation,
	Ducourthial \textit{et al.} in \cite{34} proposed an adaptive ND algorithm to timely detect the neighbors without consuming too much network resources in the VANET with high node density, high packet collision, and high packet loss rate.
	This algorithm relied on cooperative packet loss estimation, whose operation lacked independence and stability.
	
	However, the above algorithms can not ensure the completeness of ND and can not guarantee short-term convergence.
	This paper proves that the proposed GSIM-ND algorithm can always reach the short-term convergence quickly, which is proved by deriving the expected number of discovered neighbors within a given period.
	
\section{System Model}\label{sec_3}
	We consider a VANET consisting of RSUs and mobile nodes as shown in Fig. \ref{fig_1}(a).
	Fixed RSUs are connected with optical fiber and their locations are pre-stored in a common database accessible to all mobile nodes.
	All the RSUs and mobile nodes have unique identities.
	
	Most commercial off-the-shelf V2X communication radio devices such as DSRC are lack of ability of beamforming.
	Actually, the studies on DSRC based V2X communication with omni-directional antenna were solid and instructive for the optimization of the off-the-shelf V2X communication technologies \cite{35}-\cite{37}.
	In this paper, the directional antenna for the VANETs is applied to achieve high capacity, strong spatial diversity and large transmission range.
	
	The sensing equipment installed in the RSU includes radar.
	Compared with communication signal, radar sensing signal is more prone to fading and is more susceptible to various obstacles.
	Thus, the sensing range is smaller than the communication range as shown in Fig. \ref{fig_1}(b).
	The inconsistence of the sensing range and communication range results in the incomplete sensing information and further affects the convergence of ND.
	Through the coordination of various sensing devices and the adjustment of equipment power,
	the sensing radius of RSU is defined by $r$, which is small than the communication radius of RSU and larger than the road width $d$.
	The distance between two neighboring RSUs, defined by $s$, is greater than $2r$ when the sensing range of neighboring RSUs cannot be completely covered, i.e., there is a blind area on the road, which also results in incomplete sensing information.
	The longitudinal component of $s$ along the road is defined by $s_x$.
	As shown in Fig. \ref{fig_1}(b),
	the mobile node A within the sensing range can obtain complete sensing information from RSU.
	The mobile node B within the communication range and outside the sensing range can only obtain incomplete sensing information.
	
	To deal with the problem of incomplete sensing information,
	two different convergence conditions are set for the mobile nodes that obtain sensing information of different completeness degrees in Section \ref{sec_4}-A. 
	The mobile nodes with complete sensing information can converge more quickly.
	The mobile nodes with incomplete sensing information can still achieve convergence within a limited time, i.e., the algorithm does not fall into an infinite loop.
	
	\begin{figure}[h]
		\centering
		\includegraphics[width=0.5\textwidth]{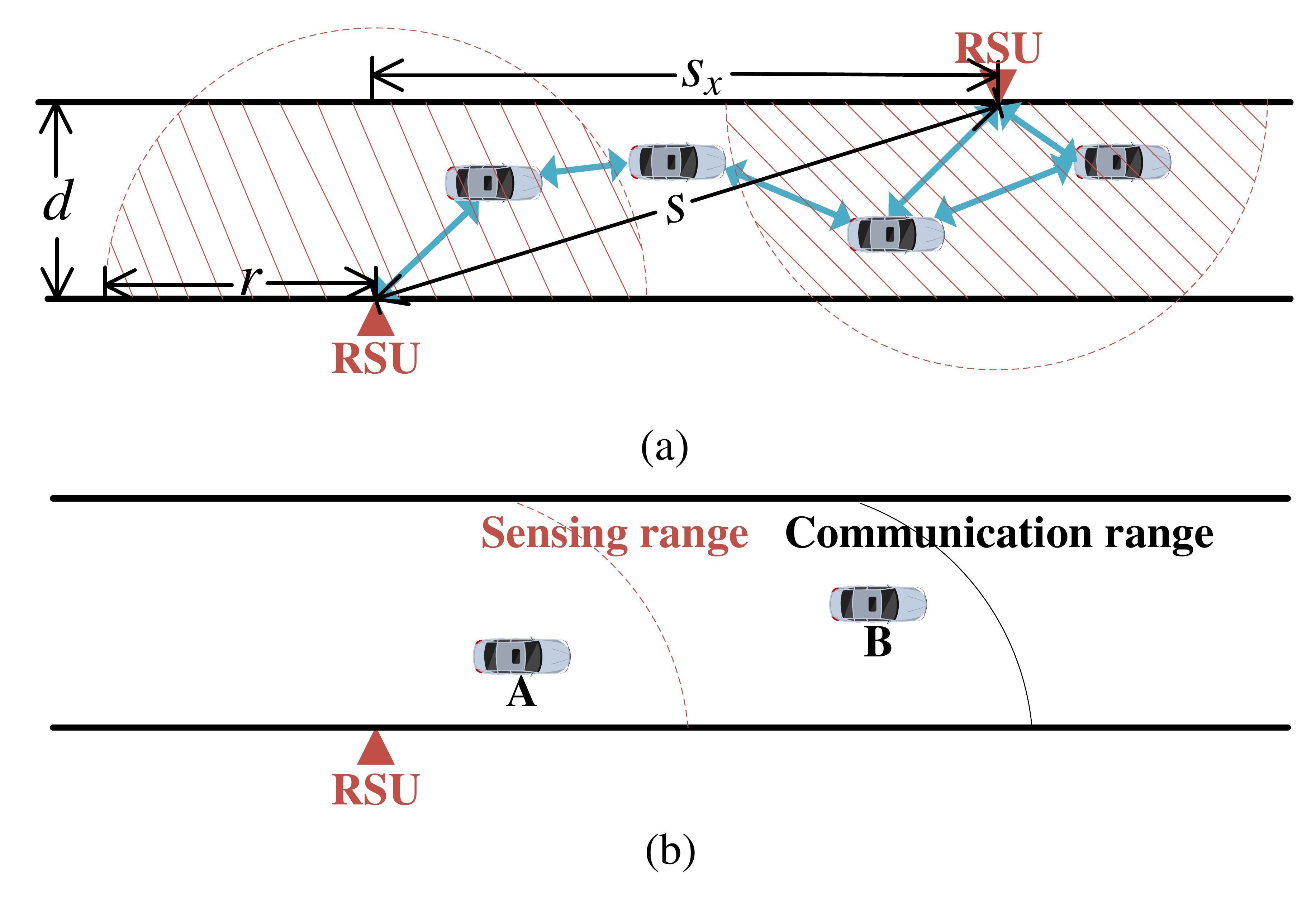}
		\caption{Fig. 1. Scenario of VANET.}
		\label{fig_1}
	\end{figure}

\section{Design of GSIM-ND Algorithm}\label{sec_4}
\subsection{Design Principles}
	This paper proposes GSIM-ND algorithm for VANET.
	Nine principles including 1) \textit{handshake}, 2) transceiver antenna, 3) time synchronization, 4) schedule of antenna beam direction, 5) schedule of transceiver state, 6) information dissemination mode, 7) information reception mode, 8) prior information, and 9) convergence conditions, are considered.
	Detailed descriptions are as follows.
	\begin{figure*}[b]
		\centering
		\includegraphics[width=1\textwidth]{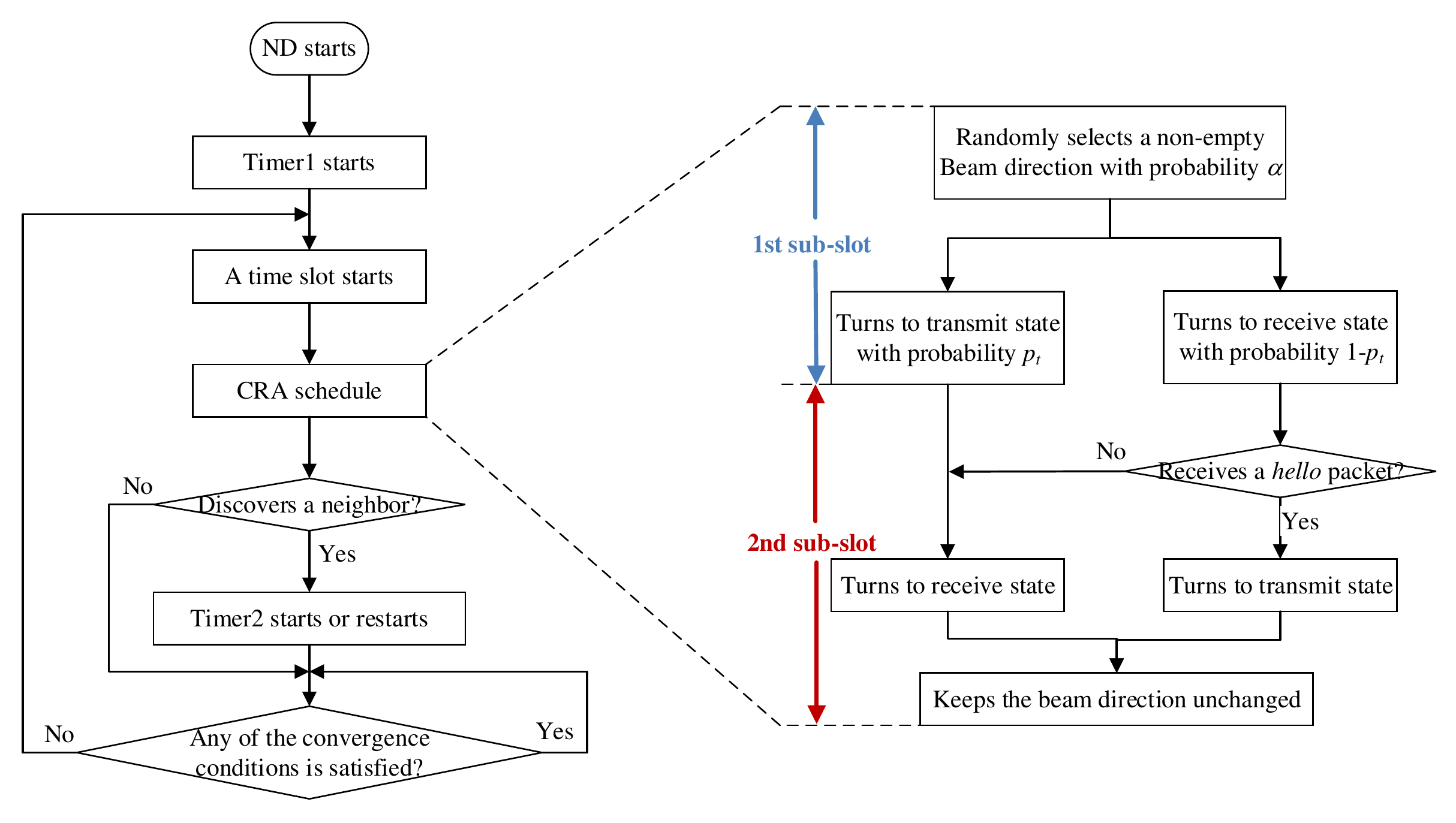}
		\caption{Fig. 3. GSIM-ND with CRA.}
		\label{fig_3}
	\end{figure*}
	\begin{itemize}
		\item [1)]\textbf{\textit{Handshake}:}
		Hybrid \textit{handshake} combining one-way and two-way \textit{handshake} methods is considered.
		A mobile node can identify a neighbor by receiving a \textit{hello} packet or a \textit{feedback} packet \cite{38}.
		\item [2)]\textbf{Transceiver Antenna:}
		The directional antenna for VANETs is applied in this paper.
		However, the omni-directional antenna can be regarded as a directional antenna with only one beam.
		The beamwidth of the directional antenna is defined by $\theta$ ($\theta \in ( 0,2\pi ]$).
		If $\theta$ is set to $2\pi$, the analysis with directional antenna could easily extend to the analysis with omni-directional antenna.
		Therefore, the designed algorithm and the analysis in this paper are still applicable to the scenario with omni-directional antenna.
		\item [3)]\textbf{Time Synchronization:}
		All the RSUs and mobile nodes are synchronized \cite{39}.
		Each time slot is divided into two sub-slots.
		The first sub-slot is used for mobile nodes to actively transmit \textit{hello} packet to neighbor mobile node.
		The second sub-slot is used for neighbor mobile node to transmit \textit{feedback} packet.
		\item [4)]\textbf{Schedule of Antenna Beam Direction:}
		For the nodes applying directional antenna, the antenna beam direction of the nodes should be scheduled for GSIM-ND.
		Detailed schedule of antenna beam direction is specified in Section \ref{sec_4}-B.
		\item [5)]\textbf{Schedule of Transceiver State:}
		Half-duplex is applied at each mobile node, i.e., a mobile node can either transmit or receive at a time.
		Detailed schedule of the transceiver state is specified in Section \ref{sec_4}-B.
		\item [6)]\textbf{Information Dissemination Mode:}
		The information in the packet determines what a mobile node can discover from
		the packet it receives, thereby affecting ND.
		In the case of low network capacity,
		each packet only contains information of the transmitter itself,
		i.e., a mobile node can only obtain its neighbors’ information and discover its neighbors by directly receiving the packets from its neighbors.
		Considering the frequent changing topology of VANET, some packets can carry more information, such as the information of the transmitter's neighbors.
		In this paper, gossip-based algorithm can be applied.
		A mobile node not only directly discovers the transmitter, but also indirectly discovers the common neighbors between the mobile node and the transmitter.		
		\item [7)]\textbf{Information Reception Mode:}
		A collision occurs when	a mobile node receives multiple packets at a time and is
		unable to demodulate each of them.
		In the case where efficient reception equipment cannot be provided, the mobile nodes can only receive one packet at a time with SPR.
		In this paper, the multi-channel communication is feasible, which enables MPR.
		In the case of MPR, the mobile nodes can receive multiple packets at a time without collisions.
		\item [8)]\textbf{Prior Information:}
		Since the prior information of the neighbors can accelerate the convergence of ND \cite{19},
		VANET with sensing function \cite{14}-\cite{18} can take sensing information as the prior information to accelerate ND.
		In this paper, RSU can obtain sensing information through sensing devices such as radar \cite{14}.
		Then, RSU broadcasts the sensing information to the road mobile nodes within its coverage area through communication devices.
		The sensing information that RSU can obtain refers to the distribution of mobile nodes within the sensing range of RSU.
		Thus, the mobile nodes equipped with OBUs can obtain the number of their neighbors within each beam when they receive the sensing information broadcast by RSUs.
		\item [9)]\textbf{Convergence Conditions:}
		For mobile nodes that can obtain the number of neighbors from the sensing information,
		ND can converge when the number of discovered neighbors reaches the number of neighbors.
		For mobile nodes that can not obtain the number of neighbors from the sensing information,
		since the unknown number of neighbors would lead to at most a two-factor deceleration in the rate of ND convergence \cite{25}, \cite{28},
		ND can converge when the time of not discovering any neighbors reaches half of the execution time of ND.
	\end{itemize}

\subsection{GSIM-ND Algorithm}
	Each mobile node has two timers, namely, timer1 and timer2.
	The timer1 starts timing when GSIM-ND algorithm starts to execute.
	The timer2 starts timing when the mobile node successfully discovers a neighbor.
	Then, timer2 restarts timing every time the mobile node successfully discovers a new neighbor.
	Thus, the GSIM-ND algorithm can converge when any of the following convergence conditions is satisfied.
	\begin{figure}[h]
		\centering
		\includegraphics[width=0.5\textwidth]{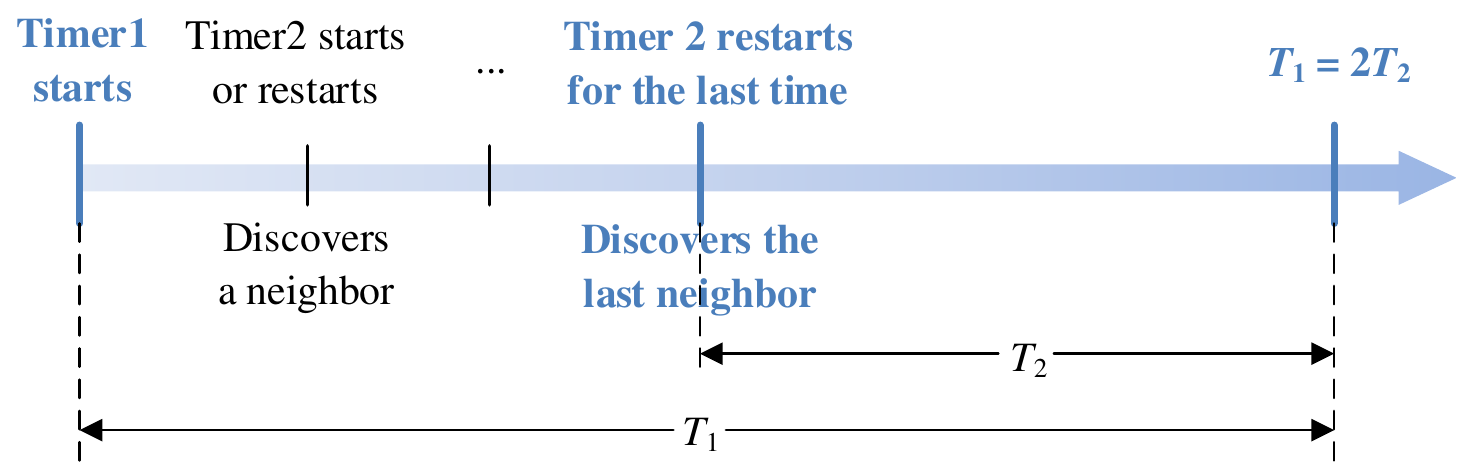}
		\caption{Fig. 2. Time axis of GSIM-ND.}
		\label{fig_2}
	\end{figure}
	\begin{itemize}
		\item [$\bullet $]
		The number of discovered neighbors reaches to the number of neighbors obtained from the sensing information for the mobile nodes with complete sensing information.
		\item [$\bullet $]
		The time by timer2, defined by $T_2$, reaches half of the time by timer1, defined by $T_1$, for the mobile nodes with incomplete sensing information.
	\end{itemize}
	
	CRA is applied to schedule the antenna beam direction and transceiver state in this paper since it is more suitable for the scenarios with high mobility compared with SBA.
	\begin{itemize}
		\item [1)]
		In the first sub-slot, a mobile node randomly selects a non-empty beam direction with probability $\alpha$.
		The mobile node either transmits a \textit{hello} packet with probability $p_{\rm t}$ or turns to receive state with probability $1-p_{\rm t}$ \cite{20}.
		\item [2)]
		In the second sub-slot, the mobile node keeps the beam direction unchanged.
		If the mobile node receives a \textit{hello} packet in the first sub-slot successfully,
		it transmits a \textit{feedback} packet in the second sub-slot.
		Otherwise, it turns to receive state in the second sub-slot \cite{20}.
	\end{itemize}
	
	The mobile node has discovered all the current neighbors when the algorithm reaches convergence.
	Since ND execution can be completed in a very short time period,
	the road mobile node can be approximately considered to be fixed during the time period of ND execution.
	However, the neighbors of the mobile nodes are likely to change as time changes after the convergence of ND.
	The re-discovery of neighbors is necessary.
	This paper focuses on the design and performance problems of the initial ND algorithm.
	With the designed algorithm executes periodically and repeatedly, the correctness of neighbor information with mobility can be guaranteed.
	
\section{Expected Number of Discovered Neighbors within a Given Period}\label{sec_5}
	RSU was only used for communication in the early stage.
	However, it gradually obtained the sensing function \cite{40}.
	For example, Siemens and Bosch have unveiled a connected vehicle collective sensing system that directly connects roadside smart cameras with dual-mode RSUs to deliver traffic and road information to connected vehicles.
	Recently, the integrated sensing and communication RSU has become a consensus \cite{41}.
	As explained in Section \ref{sec_4}-A,
	the sensing information provided by RSUs can be used as the prior information to accelerate the convergence of ND and improve the performance of ND \cite{19}.
	Therefore, to deal with the problem of ND among mobile nodes,
	the GSIM-ND algorithm that leverages the sensing information provided by RSUs is proposed.
	
	The expected number of discovered neighbors within a given period is agreed to be the metric of ND algorithms in the latest 3GPP standard \cite{42}.
	In Section \ref{sec_5},
	the expected number of discovered neighbors within a given period is theoretically derived and used as the metric to evaluate the performance of GSIM-ND algorithm.
	
	In this paper, the assistance of sensing information to ND among mobile nodes is embodied in two aspects.
	On the one hand, the convergence condition of GSIM-ND algorithm is determined by sensing information, which is explained in Section \ref{sec_4}-B.
	On the other hand, the direction and number of non-empty beams are determined by sensing information, which is explained in the derivation of Section \ref{sec_5} and Appendix A.
	Through avoid wasting time on the empty beam by only selecting non-empty beam for ND, the performance of ND can be improved.
	
\subsection{The Probability of Discovering Neighbors at an Arbitrary Time Slot}
	In this paper, we consider GSIM-ND algorithm in the case where there are $k$ different modulation modes for packet.
	Each packet selects one of the modulation modes with probability $\frac{1}{k}$.
	A mobile node can receive at most $k$ packets in each sub-slot without collisions.
	
	If a mobile node receives $n$ packets in a sub-slot ($n\leq k$), there are $k^{n}$ modulation modes for these $n$ packets, $\binom{k}{n}$ modulation modes when these $n$ packets are modulated differently.
	Thus, the probability of no collisions is
	\begin{equation}\label{eq_1}
		{P_{\rm cm}(n,k)} = \frac{\binom{k}{n}}{k^{n}}.
	\end{equation}
	
	\begin{lemma} \label{lemma1}
		The probability of a mobile node successfully receiving $n$ \textit{hello} packets in the first sub-slot is
		\begin{equation}\label{eq_2}
			{p_h(n)} = \alpha(1-p_{\rm t})(\alpha p_{\rm t})^nP_{\rm cm}(n,k)(1-\frac{n}{k}\alpha p_{\rm t})^{N_{\rm b}-n}.
		\end{equation}
		
		The probability of a mobile node successfully receiving $n$ \textit{feedback} packets in the second sub-slot is	
		\begin{equation}\label{eq_3}
			{p_f(n)} = \alpha p_{\rm t}p_{\rm tem}^{n}P_{\rm cm}(n,k)(1-\frac{n}{k}p_{\rm tem})^{N_{\rm b}-n},
		\end{equation}
		where
		\begin{equation}\label{eq_4}
			{p_{\rm tem}} = \sum\limits_{n=1}^{k}{\frac{p_h(n)}{\alpha p_{\rm t}}}.
		\end{equation}
		
		The probability of a mobile node directly discovering any neighbors at an arbitrary time slot is
		\begin{equation}\label{eq_5}
			{P_{\rm s}} = \sum\limits_{n=1}^{k}{p_{\rm s}(n)} = \sum\limits_{n=1}^{k}{p_h(n)} + \sum\limits_{n=1}^{k}{p_f(n)}.
		\end{equation}
	\end{lemma}
	\begin{proof}		
		As illustrated in Fig. \ref{fig_4}, the notations are given as follows.
		\begin{figure}[h]
			\centering
			\includegraphics[width=0.5\textwidth]{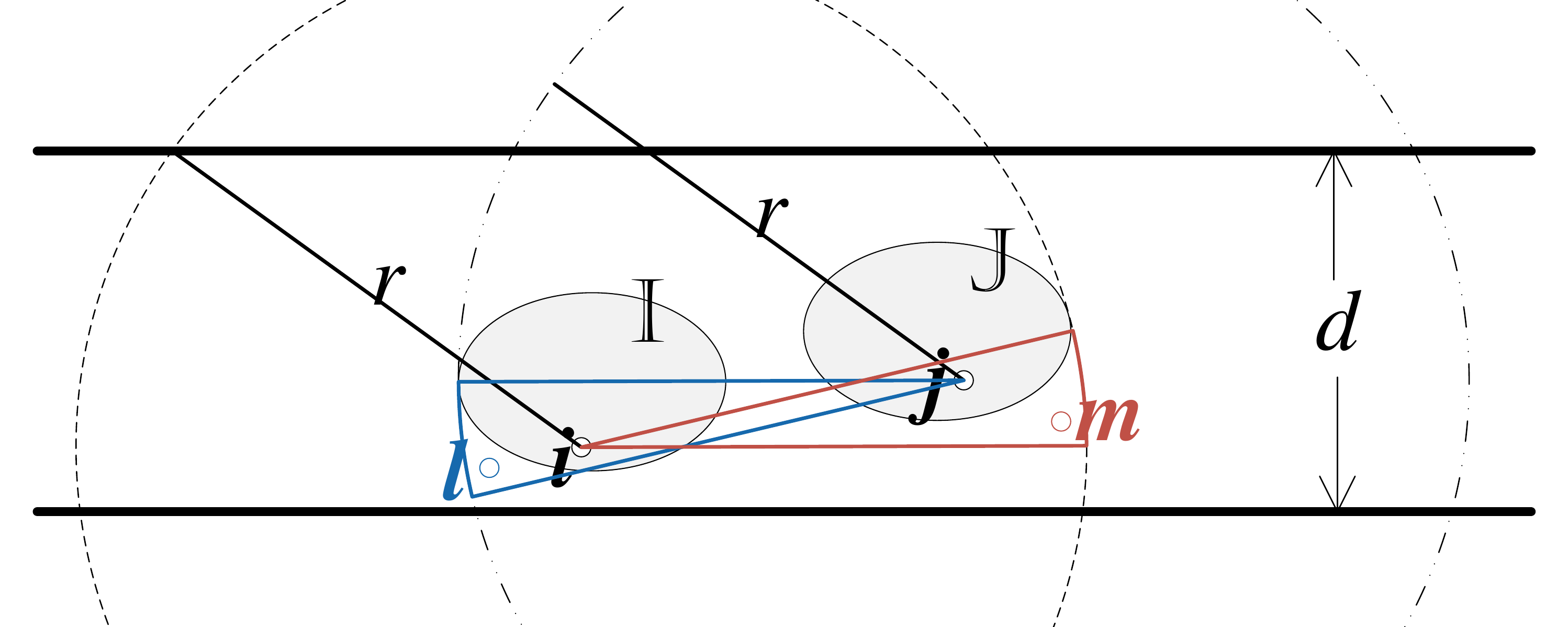}
			\caption{Fig. 4. Neighbor set of mobile nodes.}
			\label{fig_4}
		\end{figure}	
		\begin{itemize}
			\item [1)]
			Neighbor set $\mathbb{J}$: The mobile node \textit{i} has a neighbor set $\mathbb{J}$, which contains $n$ neighbors.
			\item [2)]
			Mobile node \textit{j}: The mobile node \textit{j} is an arbitrary neighbor of the mobile node \textit{i} and belongs to the neighbor set $\mathbb{J}$.
			\item [3)]
			Mobile node \textit{m}: The mobile node \textit{m} is an arbitrary neighbor of the mobile node \textit{i} in the current beam direction of the mobile node \textit{i} and does not belong to the neighbor set $\mathbb{J}$.
			\item [4)]
			Neighbor set $\mathbb{I}$: The mobile node \textit{j} has a neighbor set $\mathbb{I}$, which contains $n$ neighbors.
			The mobile node \textit{i} belongs to the neighbor set $\mathbb{I}$.
			\item [5)]
			Mobile node \textit{l}: The mobile node \textit{l} is an arbitrary neighbor of the mobile node \textit{j} in the current beam direction of the mobile node \textit{j} and does not belong to the neighbor set $\mathbb{I}$.
		\end{itemize}
		
		Then, we derive the probabilities of receiving \textit{hello} packets and \textit{feedback} packets in turn according to the GSIM-ND algorithm.
		
		\textbf{Step 1: The probability of receiving \textit{hello} packets}
		
		The probability that the mobile node \textit{i} is receiving and the beam direction is aligned with the entire neighbor set $\mathbb{J}$ in the first sub-slot is
		\begin{equation}\label{eq_6}
			{p_i} = \alpha(1-p_{\rm t}),
		\end{equation}
		where $\alpha$ is the probability of the mobile node selecting one of non-empty beam directions, $p_{\rm t}$ is the probability of the mobile node transmitting a \textit{hello} packet.
		
		The probability of the entire neighbor set $\mathbb{J}$ transmitting \textit{hello} packet with different modulation modes and the beam direction being aligned with the mobile node \textit{i} in the first sub-slot is
		\begin{equation}\label{eq_7}
			{p_{\mathbb{J}}} = (\alpha p_{\rm t})^nP_{\rm cm}(n,k).
		\end{equation}
			
		The probability of the mobile node \textit{m} not causing interference to \textit{hello} packets transmitted by the entire neighbor set $\mathbb{J}$ in the first sub-slot is
		\begin{equation}\label{eq_8}
			{p_m} = (1-\frac{n}{k}\alpha p_{\rm t})^{N_{\rm b}-n},
		\end{equation}
		where $N_{\rm b}$ is the number of neighbors in a non-empty beam of the mobile node \textit{i}.
		To calculate $P_{\rm s}$, the expectation of $N_{\rm b}$ is derived in the Appendix A.
		
		The probability of the mobile node \textit{i} successfully receiving $n$ \textit{hello} packets transmitted by the entire neighbor set $\mathbb{J}$ in the first sub-slot is $p_h(n)={p_i}{p_{\mathbb{J}}}{p_m}$.
		Therefore, $p_h(n)$ can be formulated in (\ref{eq_2}).
		
		\textbf{Step 2: The probability of receiving \textit{feedback} packets}
		
		The probability of the mobile node \textit{j} transmitting \textit{hello} packet and the beam direction being aligned with the entire neighbor set $\mathbb{I}$ in the first sub-slot is
		\begin{equation}\label{eq_9}
			{p_j} = \alpha p_{\rm t}.
		\end{equation}
		
		The probability of the entire neighbor set $\mathbb{I}$ transmitting \textit{feedback} packet with different modulation modes in the second sub-slot is
		\begin{equation}\label{eq_10}
			{p_{\mathbb{I}}} = p_{\rm tem}^nP_{\rm cm}(n,k),
		\end{equation}
		where $p_{\rm tem}$ is show in (\ref{eq_4}).
		
		The probability of the mobile node \textit{l} not causing interference to \textit{feedback} packets transmitted by the entire neighbor set $\mathbb{I}$ in the second sub-slot is
		\begin{equation}\label{eq_11}
			{p_l} = (1-\frac{n}{k}p_{\rm tem})^{N_{\rm b}-n}.
		\end{equation}
		
		The probability of the mobile node \textit{j} successfully receiving $n$ \textit{feedback} packets transmitted by the entire neighbor set $\mathbb{I}$ in the second sub-slot is $p_f(n)={p_j}{p_{\mathbb{I}}}{p_l}$.
		Therefore, $p_f(n)$ can be formulated in (\ref{eq_3}).
				
		\textit{Hello} packet and \textit{feedback} packet both carry the identity information of the transmitter.
		Therefore, $P_{\rm s}$ can be formulated in (\ref{eq_5}).
	\end{proof}
	
	The general formula for the number of time slots required when $n$ neighbors are discovered is $\sum\limits_{v=1}^{n}{\frac{1}{(n-v+1)p}}$,
	where $p$ is the probability of neighbors being discovered in a time slot \cite{43}.
	For GSIM-ND algorithm proposed in this paper, $p$ is not a fixed value,
	but a value that changes with the number of time slots.
	Therefore, there exists an upper bound and a lower bound on the number of required time slots.
	The mobile node has no discovered neighbors when the algorithm just starts to execute.
	Thus, the mobile node cannot discover neighbors indirectly, but can only discover neighbors directly.
	In this case, the probability of neighbors being discovered in a time slot is $P_{\rm s}$,
	which corresponds to the upper bound on the number of required time slots.
	The expected upper bound of the number of time slots when $n$ neighbors are directly or indirectly discovered is \cite{43}
	\begin{equation}\label{eq_12}
		\overline {t_{\rm u}(n)} = \sum\limits_{v=1}^{n}{\frac{1}{(n-v+1)P_{\rm s}}}.
	\end{equation}
	
\subsection{The Probability of Discovering Neighbors at the $t$-th Time Slot}
	In this paper, \textit{feedback} packets carry the identity information of both transmitter and transmitter's neighbors.
	However, \textit{hello} packets only carry identity information of the transmitter itself to reduce the traffic load of the network, which is different from the assumption that all the packets carry the information of transmitter's neighbors in \cite{30}.
	
	\begin{lemma} \label{lemma2}
		The probability of a mobile node directly or indirectly discovering $n$ neighbors of $n_{\rm d}$ undiscovered neighbors at the $t$-th time slot is
		\begin{figure*}[b]
			\hrulefill
			\setcounter{equation}{19}
			\begin{equation}\label{eq_20}
				{P(n,t)} = 
				\left\{
				\begin{array}{l}
					\prod\limits_{u=1}^{t}{\left(1-\sum\limits_{w=1}^{k}{P_{\rm c}(N,w,u)}\right)}, n=0\\
					P(0,t-1)P_{\rm c}(N,1,t)+P(1,t-1)\left(1-\sum\limits_{w=1}^{k}{P_{\rm c}(N-1,w,t)}\right), n=1\\
					P(0,t-1)P_{\rm c}(N,n,t)+\sum\limits_{w=1}^{n-1}{\left(P(w,t-1)P_{\rm c}(N-w,n-w,t)\right)}+P(n,t-1)\left(1-\sum\limits_{w=1}^{k}{P_{\rm c}(N-n,w,t)}\right), 1<n\leq k\\
					\sum\limits_{w=1}^{k}{\left(P(n-w,t-1)P_{\rm c}(N-n+w,w,t)\right)}+P(n,t-1)\left(1-\sum\limits_{w=1}^{k}{P_{\rm c}(N-n,w,t)}\right), k<n<N\\
					\sum\limits_{w=1}^{k}{\left(P(N-w,t-1)P_{\rm c}(w,w,t)\right)}+P(N,t-1), n = N
				\end{array}
				\right.
			\end{equation}
		\end{figure*}
		\setcounter{equation}{12}
		\begin{equation}\label{eq_13}
			{P_{\rm c}(n_{\rm d},n,t)} = \frac{\binom{n_{\rm d}}{n}\cdot p_{\rm gs}(n,t)}{\sum\limits_{n=1}^{k}{\left [ \binom{n_{\rm d}}{n}\cdot p_{\rm gs}(n,t) \right ]}+(1-P_{\rm gs}(t))},
		\end{equation}
		where
		\begin{equation}\label{eq_14}
			p_{\rm gs}(n,t) = p_{\rm s}(n)+(1-P_{\rm s})A(n,t),
		\end{equation}
		\begin{equation}\label{eq_15}
			{P_{\rm gs}(t)} = \sum\limits_{n=1}^{k}{p_{\rm gs}(n,t)} = P_{\rm s} + (1-P_{\rm s})\sum\limits_{n=1}^{k}{A(n,t)}.
		\end{equation}
	\end{lemma}
	
	\begin{proof}
		As illustrated in Fig. \ref{fig_5}, the notations are given as follows.
		\begin{figure}[h]
			\centering
			\includegraphics[width=0.5\textwidth]{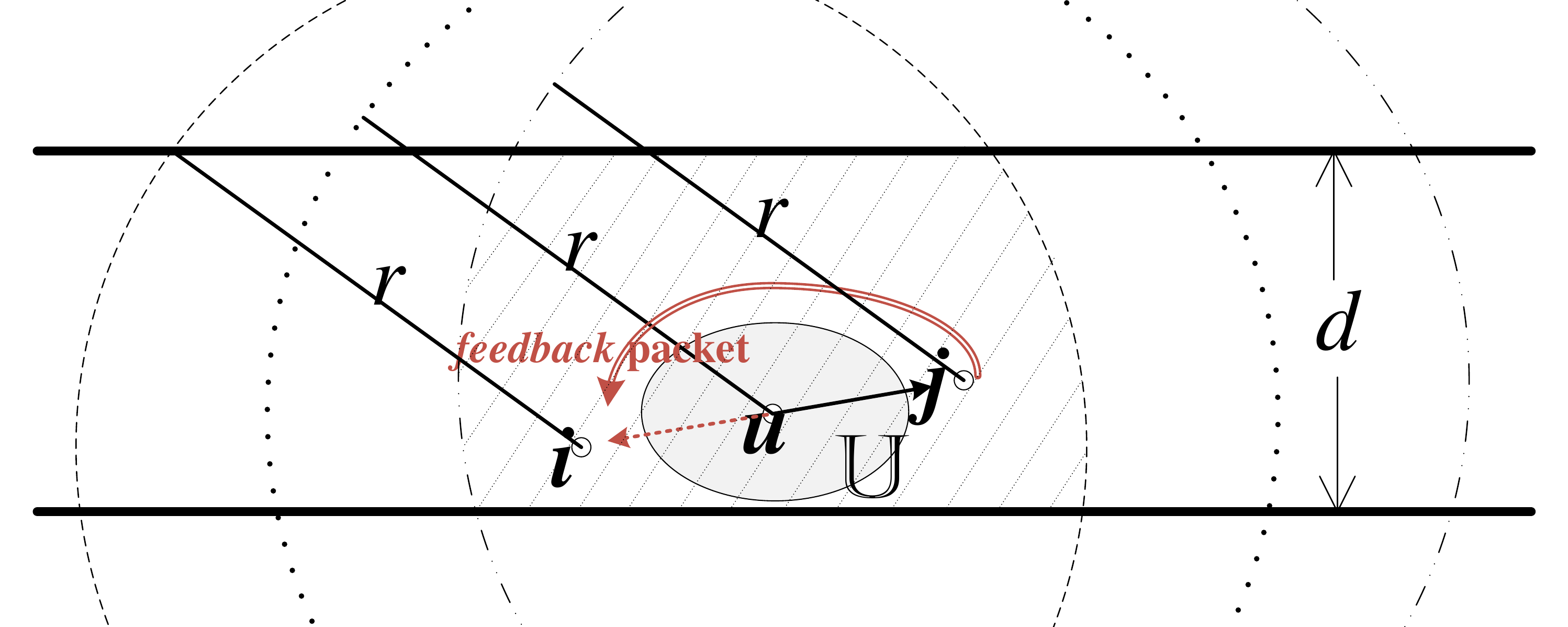}
			\caption{Fig. 5. Common neighbor set of mobile nodes.}
			\label{fig_5}
		\end{figure}	
		\begin{itemize}
			\item [1)]
			Neighbor set $\mathbb{U}$: The mobile node \textit{i} and \textit{j} have a common neighbor set $\mathbb{U}$, including $n$ neighbors of \textit{i} and \textit{j}.
			\item [2)]
			Mobile node \textit{u}: The mobile node \textit{u} is an arbitrary node that belongs to the neighbor set $\mathbb{U}$.
			The mobile node \textit{u} has been discovered by the mobile node \textit{j}.
		\end{itemize}
		
		Therefore, the mobile node \textit{i} can not only directly discovers \textit{j} but also indirectly discovers \textit{u} when \textit{i} receives a \textit{feedback} packet from \textit{j}.
	
		The probability of the mobile node \textit{i} directly discovering a neighbor \textit{j} within $t$ time slots is \cite{30}
		\begin{equation}\label{eq_16}
			{D(t)} = 1-(1-P_{\rm s})^t.
		\end{equation}
	
		The probability of the mobile node \textit{i} indirectly discovering a neighbor \textit{u} within $t$ time slots is defined as $I(t)$.
		The probability of the mobile node \textit{i} indirectly discovering the entire neighbor set $\mathbb{U}$ at the $t$-th time slot is defined as $A(n,t)$.
		Then, $I(t)$ and $A(n,t)$ are formulated as \cite{30}
		\begin{equation}\label{eq_17}
			{I(t)} = I(t-1)+\left(1-I(t-1)\right)\sum\limits_{n=1}^{k}{A(n,t)},
		\end{equation}		
		\begin{equation}\label{eq_18}
			{A(n,t)} = N_{\rm I}\left(D(t-1)+\left(1-D(t-1)\right)I(t-1)\right)p_f(n),
		\end{equation}
		where $N_{\rm I}$ is the number of common neighbors between the mobile node \textit{i} and mobile node \textit{j}, $I(1)=0$, and $A(n,1)=0$ \cite{30}.
		To calculate $P_{\rm gs}$, the expectation of $N_{\rm I}$ is derived in the Appendix B.
		
		The mobile nodes either discover their neighbors directly or indirectly.
		Therefore, the probability of a mobile node directly or indirectly discovering $n$ neighbors at the $t$-th time slot can be formulated in (\ref{eq_14}).
		The probability of a mobile node directly or indirectly discovering any neighbors at the $t$-th time slot can be formulated in (\ref{eq_15}).
		Supposing that there are $n_{\rm d}$ undiscovered neighbors of a mobile node ($n_{\rm d} \geq n$), $P_{\rm gs}(n,t)$ can be normalized to $P_{\rm c}(n_{\rm d},n,t)$ in (\ref{eq_13}).
	\end{proof}
	
	As more and more neighbors have been discovered by mobile nodes,
	the mobile nodes can gradually accelerate ND by indirectly discovering neighbors.
	In this case, the probability of neighbors being discovered in a time slot is $P_{\rm gs}$.
	$P_{\rm gs}$ gradually increases to a stable value and then remains unchanged.
	$P_{\rm gs}(\infty)$ corresponds to the lower bound on the number of required time slots.
	The expected lower bound of the number of time slots when $n$ neighbors need to be directly or indirectly discovered is \cite{43}
	\begin{equation}\label{eq_19}
		\overline {t_{\rm l}(n)} = \sum\limits_{v=1}^{n}{\frac{1}{(n-v+1)P_{\rm gs}(\infty)}}.
	\end{equation}
	
	The mobile node can only discover its neighbors directly when the gossip mechanism is not adopted.
	In this case, the upper bound is reachable.
	Since the mobile node has no neighbors already discovered at the beginning,
	it can not indirectly discover new neighbors.
	The gossip mechanism does not make sense when the algorithm has just started.
	Thus, the lower bound can only be approached and cannot be reached.
	
\subsection{The Probability of Discovering Neighbors within $t$ Time Slots}
	 The number of the mobile node's neighbors is $N$.
	\begin{lemma} \label{lemma3}
		The probability of a mobile node directly or indirectly discovering $n$ neighbors of all the undiscovered neighbors within $t$ time slots is formulated in (\ref{eq_20}).
				
		The expected number of directly or indirectly discovered neighbors within $t$ time slots is
		\setcounter{equation}{20}
		\begin{equation}\label{eq_21}
			{\overline {n(t)}}=\sum\limits_{v=1}^{N}{vP(v,t)}.
		\end{equation}
	\end{lemma}
	\begin{proof}
		We divide the discussion into the following cases.
		
		\textbf{Case 1: $n=0$}
		
		The probability of a mobile node not discovering any neighbors of all the undiscovered neighbors at the $u$-th time slot is $P_u = 1-\sum\limits_{w=1}^{k}{P_{\rm c}(N,w,u)}$.
		The probability of a mobile node not discovering any neighbors of all the undiscovered neighbors within $t$ time slots is $\prod\limits_{u=1}^{t}{P_u}$.
		
		\textbf{Case 2: $n=1$}
		
		When no neighbor is discovered within $t-1$ time slots, the probability of discovering one neighbor at the $t$-th time slot is $P_{\rm c}(N,1,t)$.
		
		When one neighbor is discovered within $t-1$ time slots, the probability of not discovering any neighbors at the $t$-th time slot is $1-\sum\limits_{w=1}^{k}{P_{\rm c}(N-1,w,t)}$.
		
		\textbf{Case 3: $1<n\leq k$}
		
		When no neighbor is discovered within $t-1$ time slots, the probability of discovering $n$ neighbors of all the undiscovered neighbors at the $t$-th time slot is $P_{\rm c}(N,n,t)$.
		
		When $w$ neighbors are discovered within $t-1$ time slots, the probability of discovering $n-w$ neighbors of the $N-w$ undiscovered neighbors at the $t$-th time slot is $P_{\rm c}(N-w,n-w,t)$.
		
		When $n$ neighbors are discovered within $t-1$ time slots, the probability of not discovering any neighbors at the $t$-th time slot is $1-\sum\limits_{w=1}^{k}{P_{\rm c}(N-n,w,t)}$.
		
		\textbf{Case 4: $k < n < N$}
						
		When $n-w$ neighbors are discovered within $t-1$ time slots, the probability of discovering $w$ neighbors of the $N-n+w$ undiscovered neighbors at the $t$-th time slot is $P_{\rm c}(N-n+w,w,t)$.
		
		When $n$ neighbors are discovered within $t-1$ time slots, the probability of not discovering any neighbors at the $t$-th time slot is $1-\sum\limits_{w=1}^{k}{P_{\rm c}(N-n,w,t)}$.
		
		\textbf{Case 5: $n = N$}
		
		When $N-w$ neighbors are discovered within $t-1$ time slots, the probability of discovering $w$ neighbors of the $w$ undiscovered neighbors at the $t$-th time slot is $P_{\rm c}(w,w,t)$.
		
		When $N$ neighbors are discovered within $t-1$ time slots, the probability of not discovering any neighbors at the $t$-th time slot is $1$.
		
		Overall, $P(n,t)$ can be formulated in (\ref{eq_20}).
		
		To calculate $\overline {n(t)}$, the expectation of $N$ is derived in the Appendix B.
	\end{proof}
	
\section{Simulation Results and Analysis}\label{sec_6}
	In this section,
	we randomly scatter vehicles over a finite area.
	The values of main simulation parameters are shown in Table \ref{table2}.
	\begin{table}[h]
		\centering
		\renewcommand\arraystretch{1.5}
		\caption{Table II\\Main Simulation parameters}\label{table2}
		\begin{tabular}{p{2cm}<{\centering}p{1cm}<{\centering}|p{2cm}<{\centering}p{1cm}<{\centering}}
			\hline
			Parameter&Value&Parameter&Value\\
			\hline
			$r$&
			$200$ ${\rm m}$&
			$s_x$&
			$600$ ${\rm m}$
			\\			
			$L$&
			$1000$ ${\rm m}$&
			$d$&
			$60$ ${\rm m}$
			\\
			$\theta$&
			$\frac{\pi}{6}$&
			$B$&
			$12$
			\\
			$p_{\rm t}$&
			$[0.1,0.9]$&
			$M$&
			$[50,1000]$
			\\
			\hline
		\end{tabular}
	\end{table}
	\begin{figure}[h]
		\centering
		\includegraphics[width=0.5\textwidth]{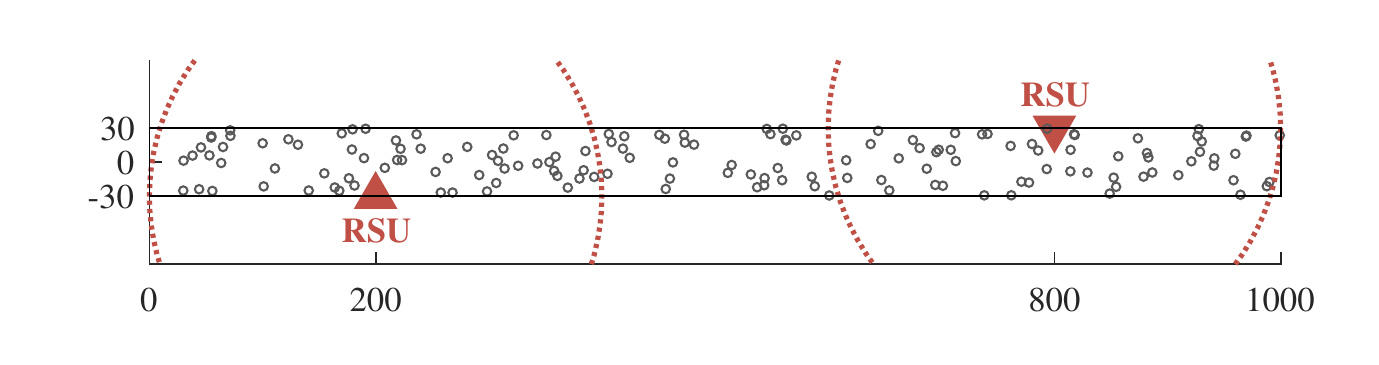}
		\caption{Fig. 6. Randomly distribution of vehicles when $M=150$.}
		\label{fig_6}
	\end{figure}
	\begin{figure}[!t]
		\centering
		\includegraphics[width=0.5\textwidth]{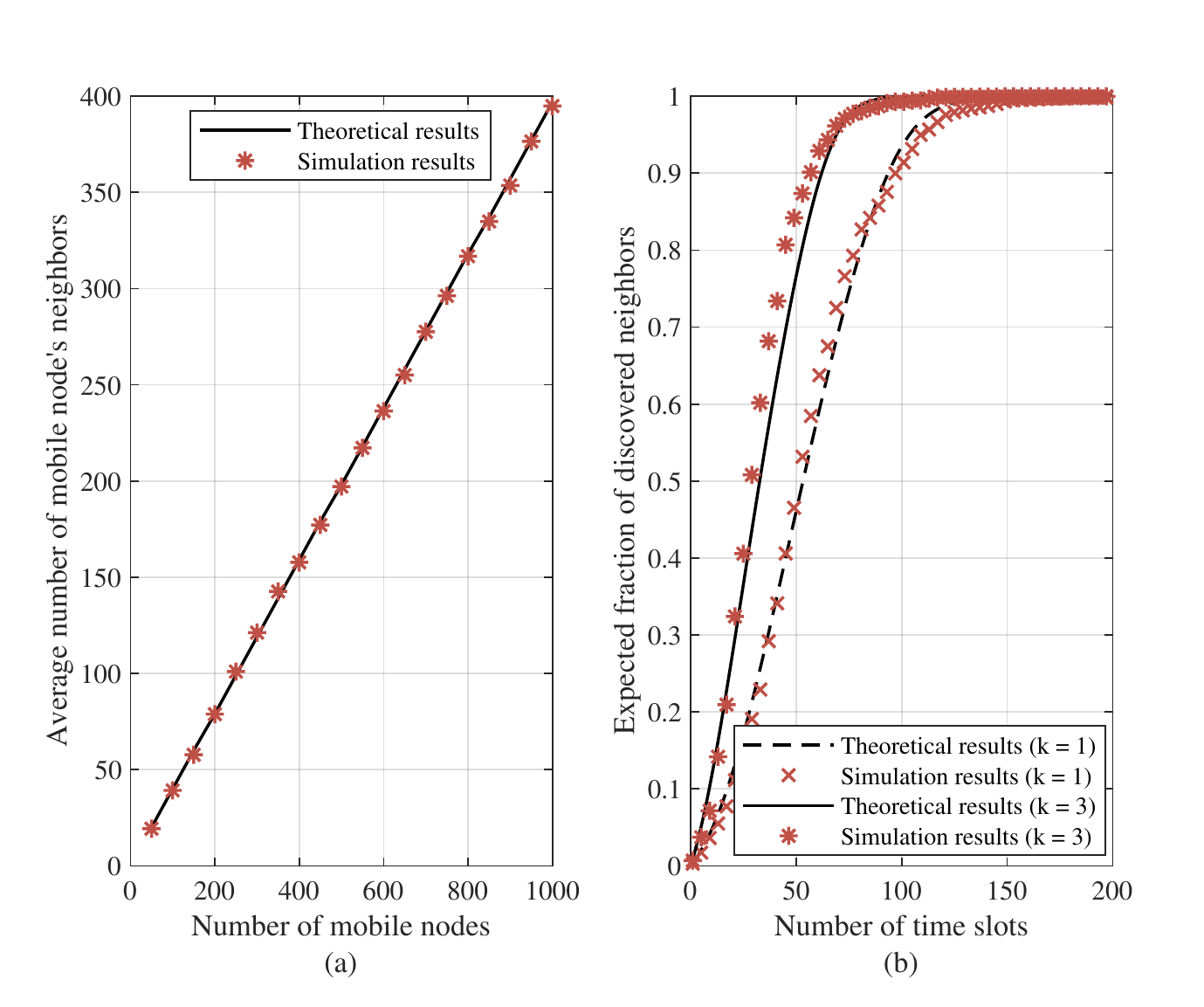}
		\caption{Fig. 7. Theoretical results and simulation results with GSIM-ND algorithm when $M=150$.}
		\label{fig_7}
	\end{figure}
	
\subsection{Verification of Theoretical Derivation}
	\begin{figure}[t]
		\centering
		\includegraphics[width=0.5\textwidth]{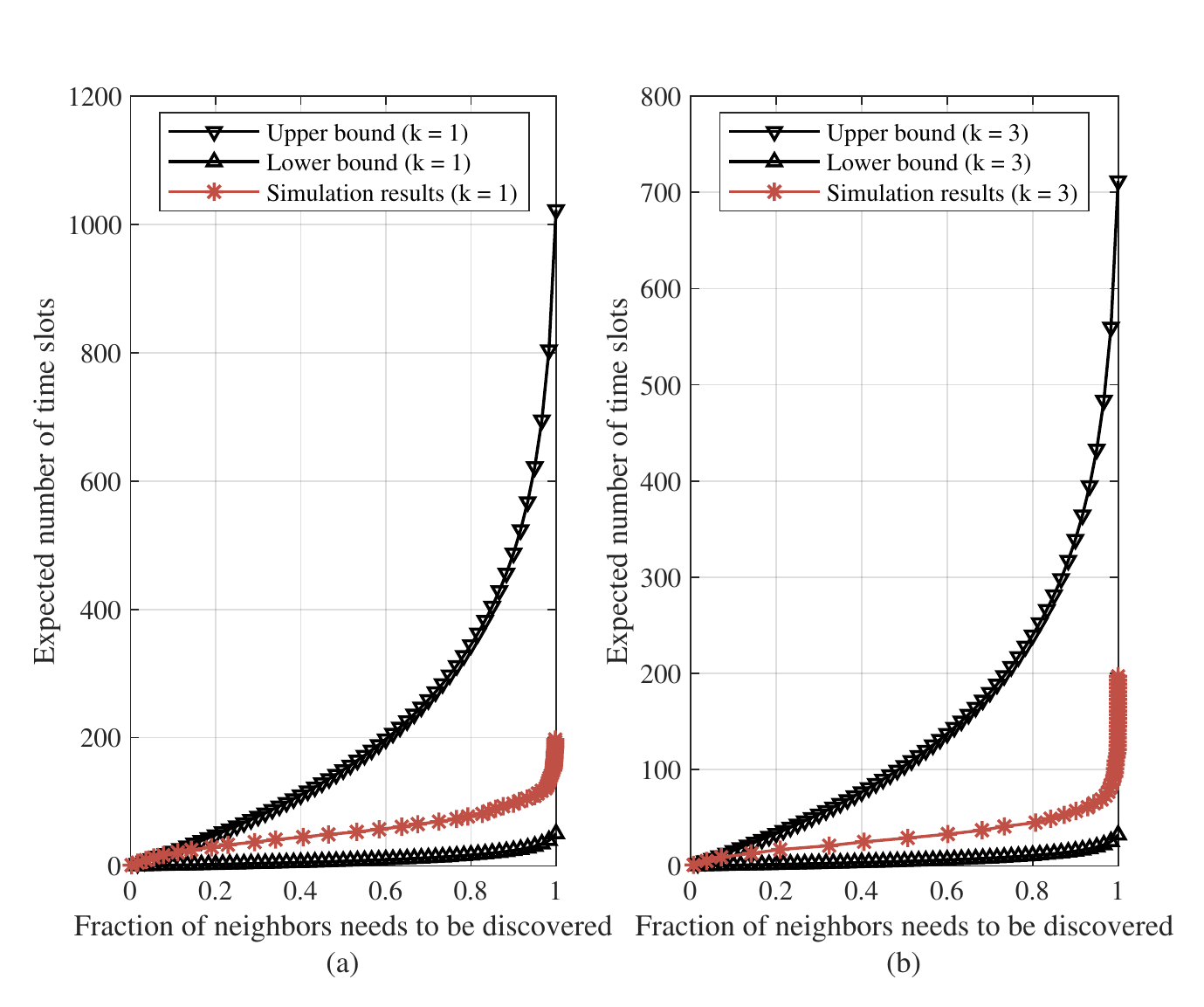}
		\caption{Fig. 8. Simulation results of lower bound and upper bound with GSIM-ND algorithm when $M=150$.}
		\label{fig_8}
	\end{figure}
	
	The distribution of mobile nodes in VANET randomly generated by Monte Carlo simulation is shown in Fig. \ref{fig_6} where the number of mobile nodes $M$ is $150$.
	The small black circles represent the vehicles in VANET, the red triangles represent the RSUs, the large circles with a red dotted line represents the sensing range of RSUs.

	The simulation results of $\overline{N}$ depend on the distribution of random scattered points.
	The theoretical results of the average number of neighbors $\overline{N}$ depend on the average density of mobile nodes and are verified in Fig. \ref{fig_7}(a).
	Since the theoretical results are very close to the simulation results, the correctness of the theoretical results is verified.
	
	In Section \ref{sec_5}, we derive the relation between the expected fraction of discovered neighbors and the number of time slots with GSIM-ND algorithm.
	Fig. \ref{fig_7}(b) shows the theoretical and simulation results of the expected fraction of discovered neighbors versus the number of time slots with GSIM-ND algorithm when the number of mobile nodes $M$ is $150$ and the number of modulation modes $k$ is $1$ and $3$ respectively.
	The theoretical results reveal the variation feature of the simulation results, which proves the rationality and correctness of theoretical derivation.
	
	The relation between the expected bounds of the number of time slots and the fraction of neighbors needs to be discovered with GSIM-ND algorithm is derived in Section \ref{sec_5}.
	Fig. \ref{fig_8} shows the simulation results and theoretical results of the expected number of time slots versus the fraction of neighbors needs to be discovered with GSIM-ND algorithm when the number of mobile nodes $M$ is $150$ and the number of modulation modes $k$ is $1$ and $3$ respectively.
	The simulation results are always between the theoretical lower bound and upper bound, which proves the rationality and correctness of theoretical derivation.
	
	Then, we analyze the efficiency of GSIM-ND based on theoretical results.
	
\subsection{Efficiency of GSIM-ND Algorithm}
	CRA algorithm and SBA algorithm are two classic ND algorithms.
	Many existing researches on ND are proposed on the basis of CRA algorithm and SBA algorithm.
	Gossip-based algorithm can greatly accelerate the convergence of ND.
	Gossip-based algorithm is one of the state-of-the-art (SOTA) methods.
	In this paper, GSIM-ND algorithm is proposed on the basis of CRA schedule and gossip mechanism.
	Therefore, GSIM-ND algorithm is compared with CRA algorithm, SBA algorithm and gossip-based algorithm in this section respectively.
	Moreover, we also compare GSIM-ND algorithms under different parameters of $k$, 
	where the parameter $k$ represents the maximum number of packets that can be received simultaneously.
	
	\begin{figure}[h]
		\centering
		\includegraphics[width=0.5\textwidth]{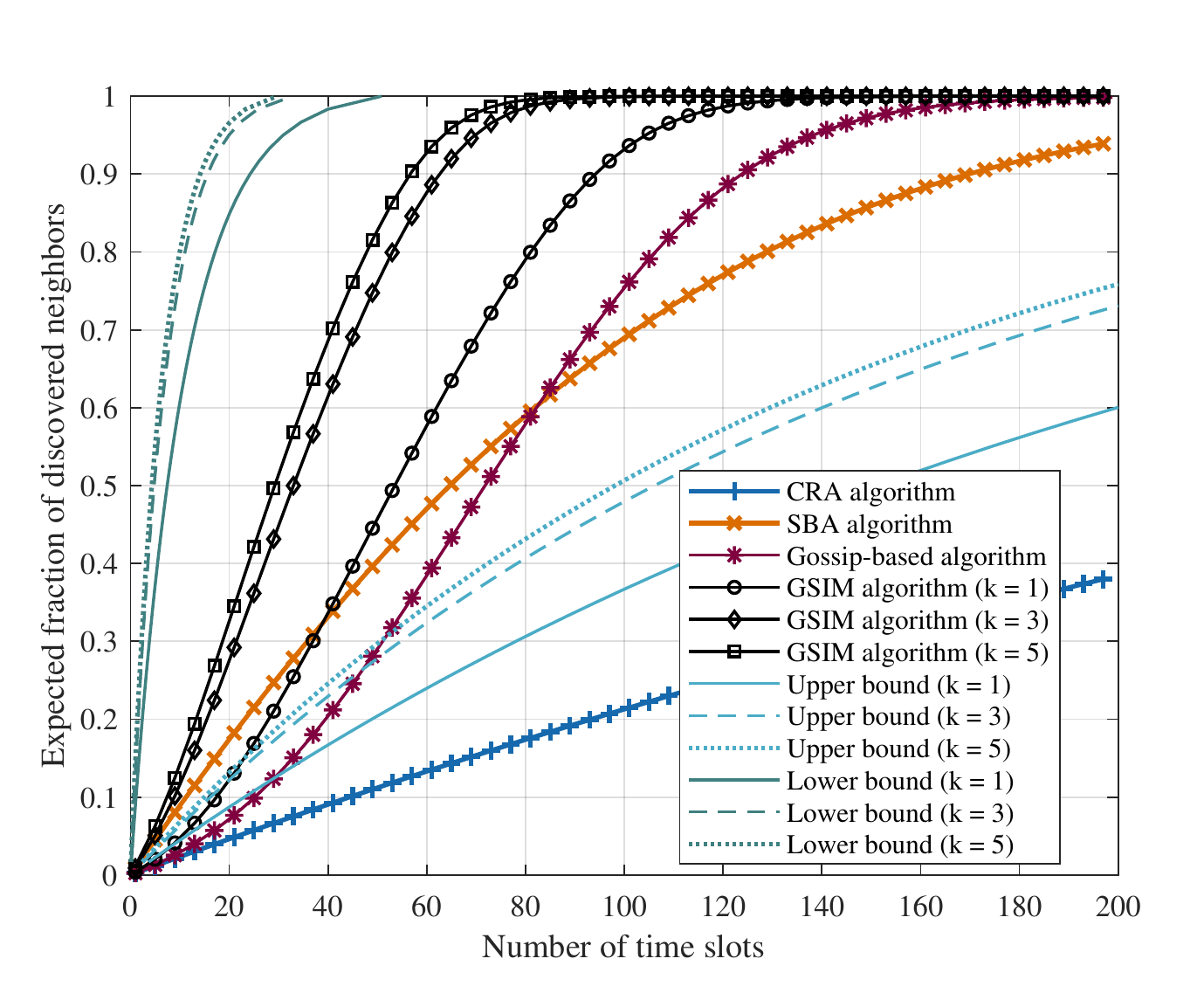}
		\caption{Fig. 9. Expected fraction of discovered neighbors versus the number of time slots with CRA algorithm, SBA algorithm, gossip-based algorithm and GSIM-ND algorithm when $M=150$.}
		\label{fig_9}
	\end{figure}
	\begin{figure}[t]
		\centering
		\includegraphics[width=0.5\textwidth]{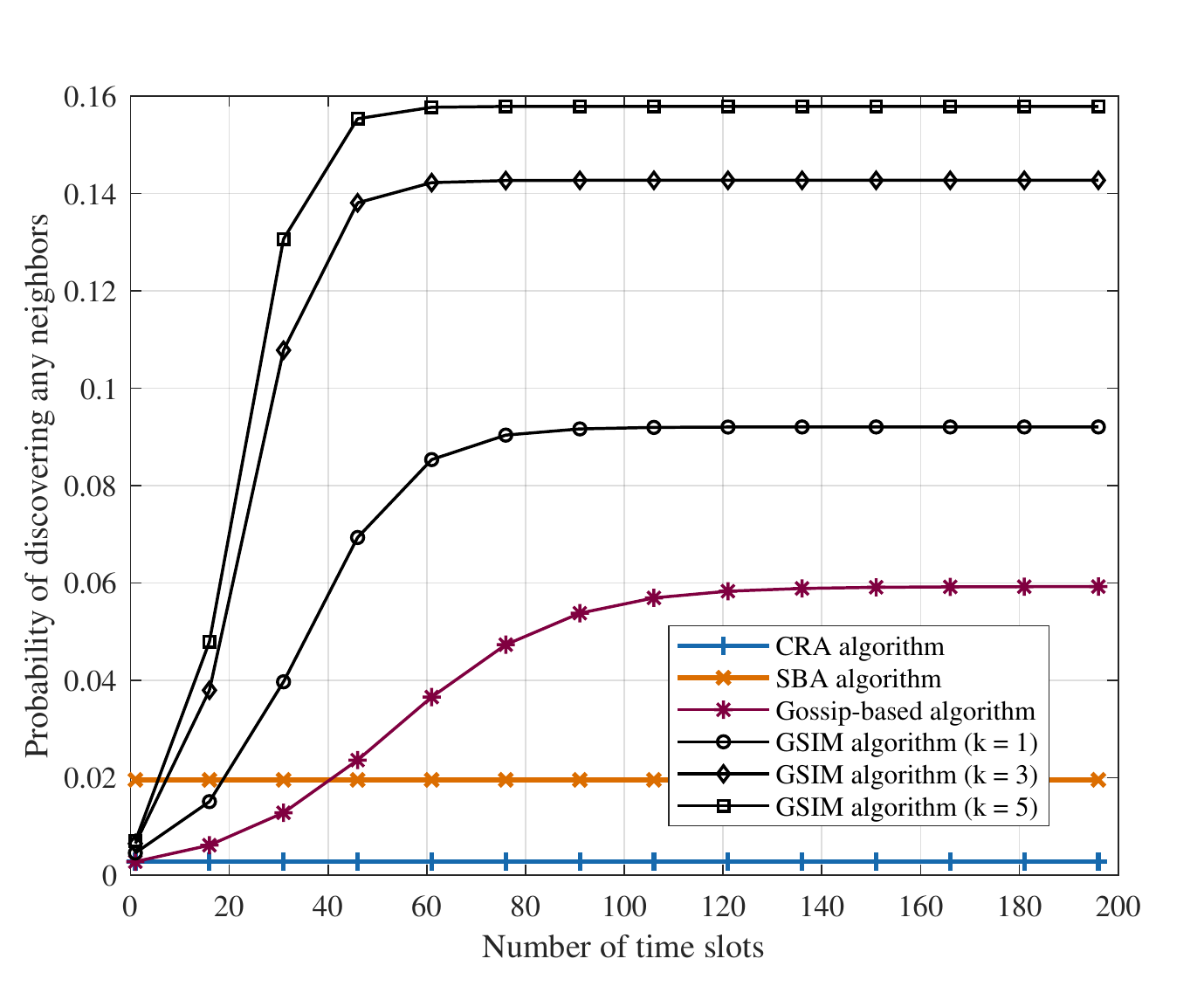}
		\caption{Fig. 10. Probability of discovering any neighbors versus the number of time slots with CRA algorithm, SBA algorithm, gossip-based algorithm and GSIM-ND algorithm when $M=150$.}
		\label{fig_10}
	\end{figure}
	\begin{figure}[t]
		\centering
		\includegraphics[width=0.5\textwidth]{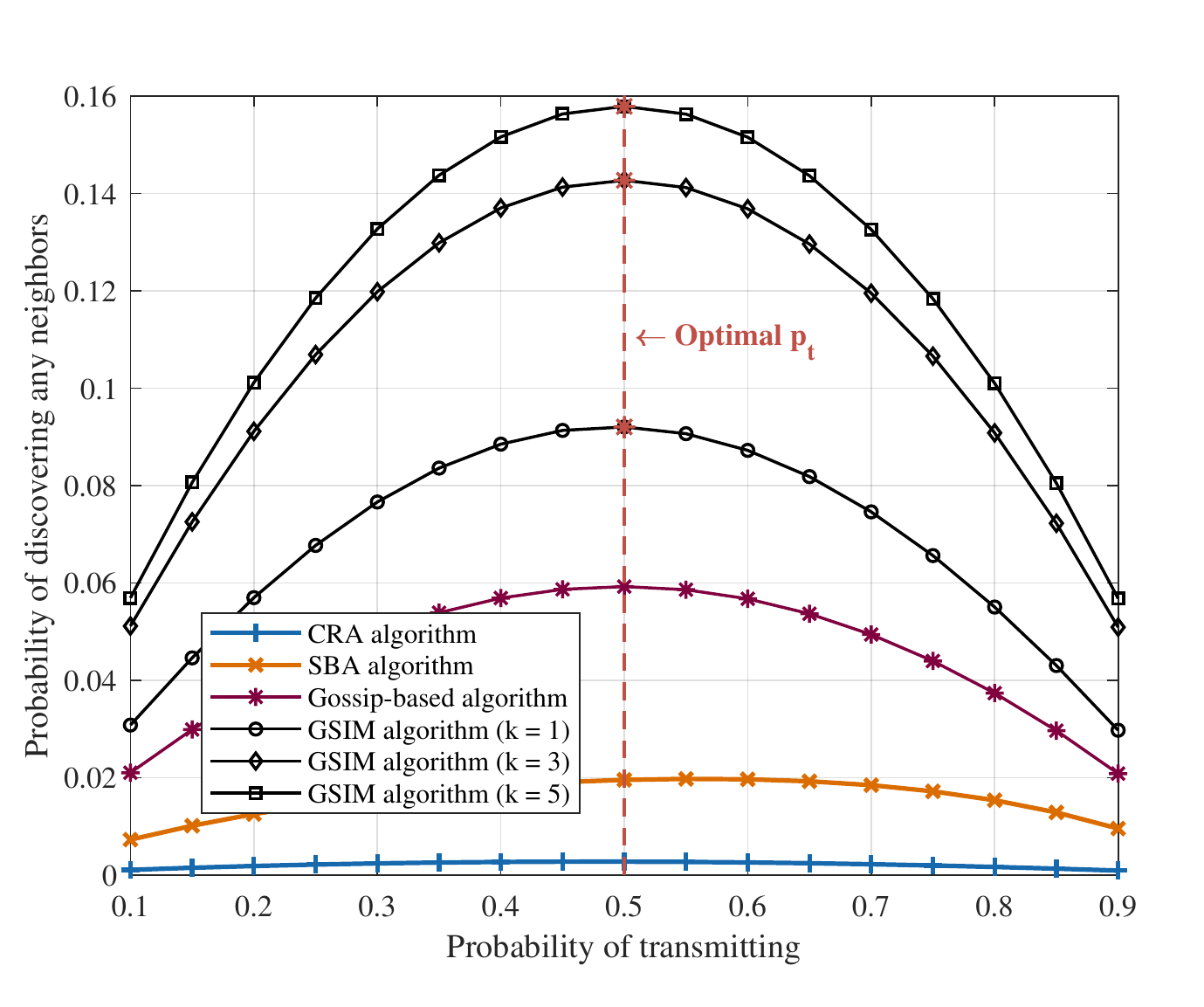}
		\caption{Fig. 11. Probability of discovering any neighbors versus the probability of a mobile node actively transmitting \textit{hello} packet with CRA algorithm, SBA algorithm, gossip-based algorithm and GSIM-ND algorithm when $M=150$.}
		\label{fig_11}
	\end{figure}
	
	Fig. \ref{fig_9} shows the expected fraction of discovered neighbors versus the number of time slots when the number of mobile nodes $M$ is $150$ and the number of modulation modes $k$ is $1$, $3$ and $5$ respectively.
	The GSIM-ND algorithm is better than CRA algorithm, SBA algorithm and gossip-based algorithm.
	The efficiency of the GSIM-ND algorithm is increasing with the increase of $k$.
		
	The factors that directly affects the efficiency of GSIM-ND algorithm, gossip-based algorithm and CRA algorithm (SBA algorithm) are the probability of discovering any neighbors $P_{\rm gs}$ and $P_{\rm s}$, respectively.
	The efficiency of GSIM-ND algorithm and gossip-based algorithm are increasing with the increase of $P_{\rm gs}$.
	The efficiency of CRA algorithm and SBA algorithm are increasing with the increase of $P_{\rm s}$.
	$P_{\rm gs}$ and $P_{\rm s}$ depend on the probability of a mobile node actively transmitting \textit{hello} packet $p_{\rm t}$ and the probability of collisions.
	
	Fig. \ref{fig_10} shows the probability of discovering any neighbors $P_{\rm gs}$ and $P_{\rm s}$ versus the number of time slots with CRA algorithm, SBA algorithm, gossip-based algorithm and GSIM-ND algorithm when the number of mobile nodes $M$ is $150$ and the number of modulation modes $k$ is $1$, $3$ and $5$ respectively.
	At the beginning of algorithms, $P_{\rm gs}$ and $P_{\rm s}$ are the same, but $P_{\rm s}$ doesn't change over time with CRA algorithm and SBA algorithm whereas $P_{\rm gs}$ increases a little bit over time and then remains stable with gossip-based algorithm and GSIM-ND algorithm.
	
	Fig. \ref{fig_11} shows the stable probability of discovering any neighbors $P_{\rm gs}$ and $P_{\rm s}$ versus the probability of a mobile node actively transmitting \textit{hello} packet $p_{\rm t}$ with CRA algorithm, SBA algorithm, gossip-based algorithm and GSIM-ND algorithm when $M$ is $150$ and the number of modulation modes $k$ is $1$, $3$ and $5$ respectively.
	When $p_{\rm t}$ is too small or too large, $P_{\rm gs}$ and $P_{\rm s}$ are small because the probability of a mobile node successfully receiving the packets from the neighbor is small.
	Fig. \ref{fig_11} proves that the optimal $p_{\rm t}$ is 0.5, which results in the maximum $P_{\rm gs}$ and $P_{\rm s}$.
	
	Hence, through such a series of comparisons, 
	the performance improvement of GSIM-ND algorithm compared to CRA algorithm, SBA algorithm and gossip-based algorithm can be intuitively revealed.
	The performance difference of GSIM-ND algorithms corresponding to different parameters $k$ can be intuitively revealed too.
	The efficiency of GSIM-ND algorithm is higher than that of CRA algorithm, SBA algorithm and gossip-based algorithm.
	As the modulation modes $k$ increases, $P_{\rm gs}$ increases slightly.
	These conclusions are consistent with those obtained from Fig. \ref{fig_9}.
	However, since a large $k$ will result in large algorithm complexity, a large $k$ is not always reasonable in practice.
	
\subsection{Stability of GSIM-ND Algorithm}
	In practical application, the density of mobile nodes in VANET is affected by many factors.
	We study the stability of GSIM-ND algorithm by analyzing the influence of the number of mobile nodes $M$ in a certain region on the efficiency of GSIM-ND algorithm.
	
	\begin{figure}[h]
		\centering
		\includegraphics[width=0.5\textwidth]{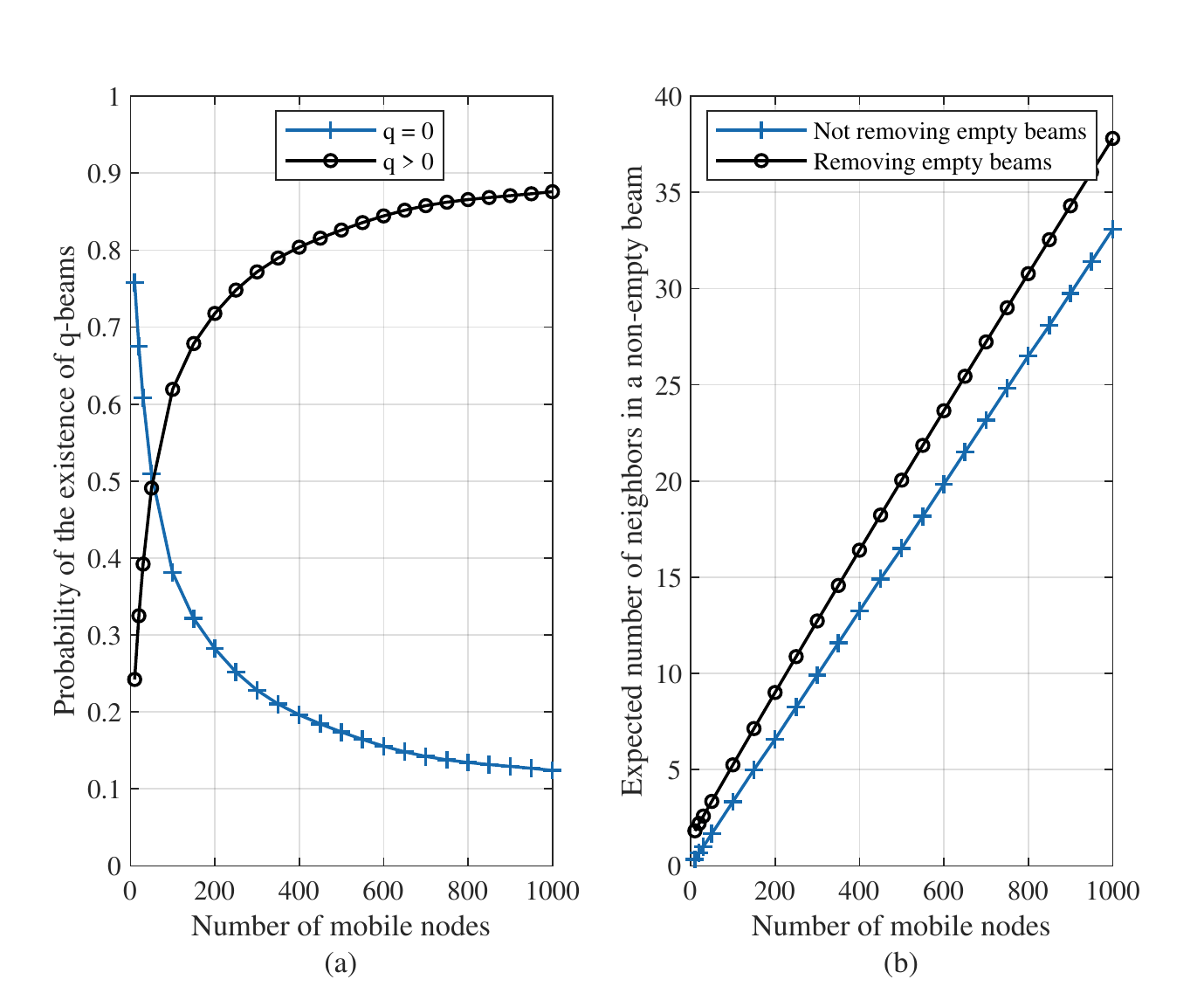}
		\caption{Fig. 12. Probability of the existence of $q$-beams and expected number of neighbors in a non-empty beam.}
		\label{fig_12}
	\end{figure}
	\begin{figure}[t]
		\centering
		\includegraphics[width=0.5\textwidth]{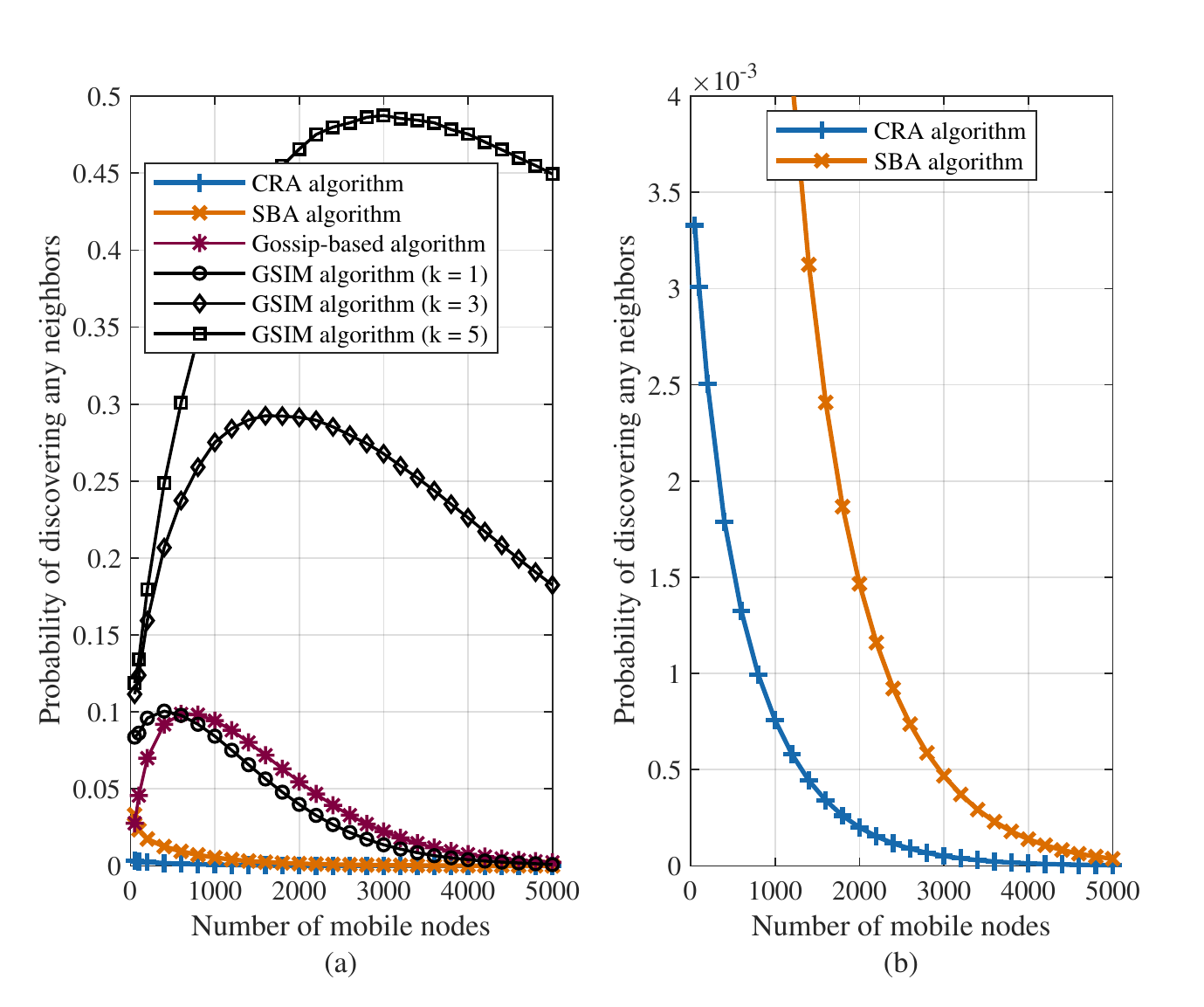}
		\caption{Fig. 13. Probability of discovering any neighbors versus the number of mobile nodes with CRA algorithm, SBA algorithm, gossip-based algorithm and GSIM-ND algorithm.}
		\label{fig_13}
	\end{figure}
	\begin{figure}[b]
		\centering
		\includegraphics[width=0.5\textwidth]{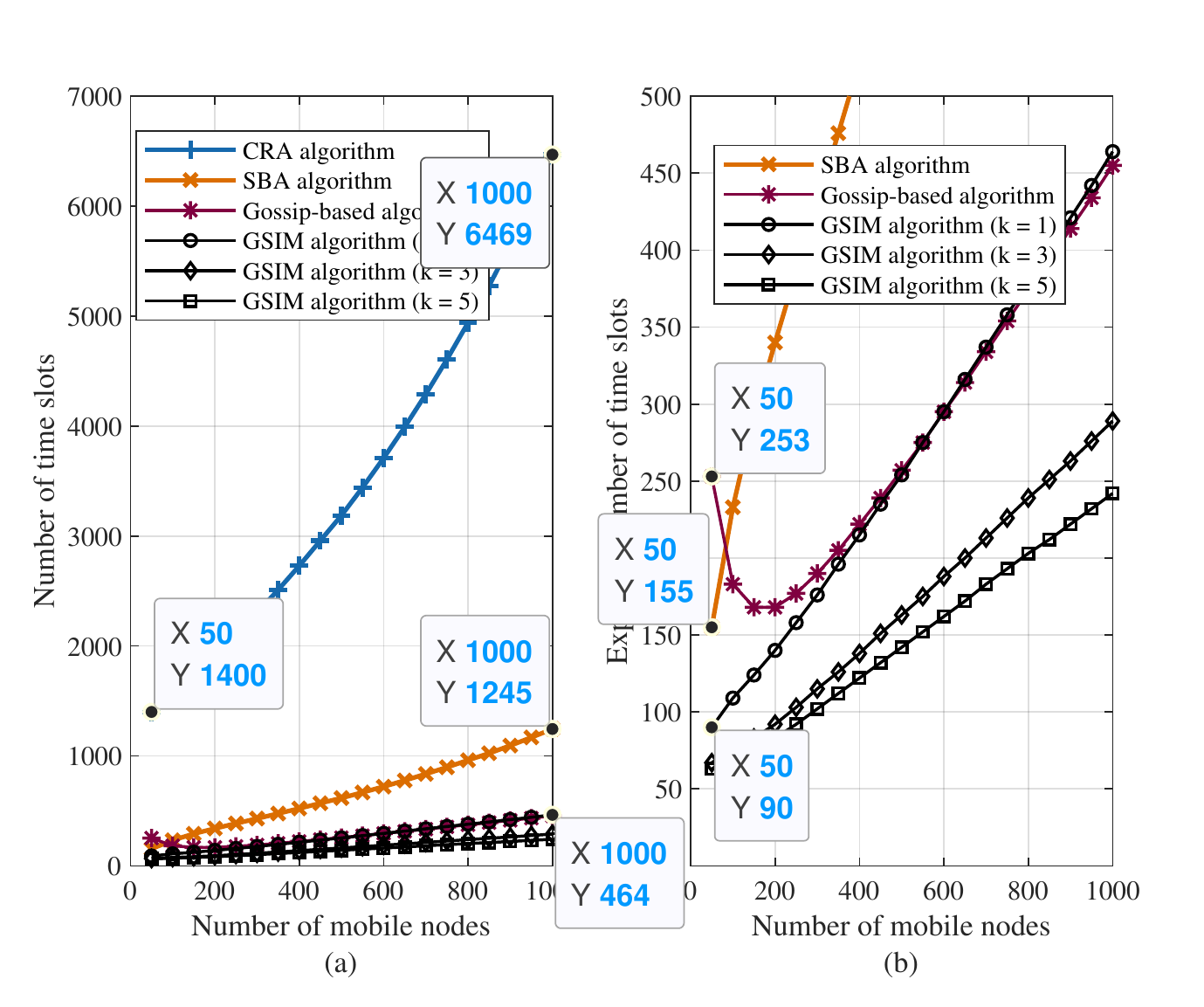}
		\caption{Fig. 14. Number of required time slots when $99\%$ neighbors are discovered versus the number of mobile nodes with CRA algorithm, SBA algorithm, gossip-based algorithm and GSIM-ND algorithm.}
		\label{fig_14}
	\end{figure}
	
	Fig. \ref{fig_12}(a) shows the probability of the existence of different kinds of beams $\frac{E_q}{B}$ versus the number of mobile nodes $M$.
	Fig. \ref{fig_12}(b) shows the expected number of neighbors in a random non-empty beam $\overline{N_{\rm b}}$ versus the number of mobile nodes $M$.
	Since a mobile node can not discover any neighbors when selecting an empty beam, the existence of empty beams has an impact on the time required for ND convergence and the efficiency of ND.
	The influence of empty beams on ND is increasing with the decrease of the number of mobile nodes $M$.
	However, the influence of empty beams on ND can not be ignored when $M$ is large.
	Through the sensing information provided by RSU, mobile nodes can only select non-empty beam for ND.
	Thus, GSIM-ND algorithm can eliminate the influence of empty beams on ND.
	
	Fig. \ref{fig_13} shows the probability of discovering any neighbors $P_{\rm gs}$ and $P_{\rm s}$ versus the number of mobile nodes $M$ with CRA algorithm, SBA algorithm, gossip-based algorithm and GSIM-ND algorithm when $p_{\rm t}$ is optimal, the number of modulation modes $k$ is $1$, $3$ and $5$ respectively.
	$P_{\rm gs}$ increases firstly and then decreases as $M$ increases with GSIM-ND algorithm and gossip-based algorithm.
	This is due to the fact that $P_{\rm gs}$ is increasing firstly with the increase of the number of common neighbors between two mobile nodes $N_{\rm I}$,
	then $P_{\rm gs}$ is decreasing because the collision is increasing with the increase of $N_{\rm I}$.
	$P_{\rm s}$ decreases as $M$ increases with CRA algorithm and SBA algorithm.
	This is due to the fact that the probability of collisions is increasing with the increase of $M$,
	which will result in the decrease of $P_{\rm s}$.
	$P_{\rm gs}$ with GSIM-ND algorithm and gossip-based algorithm are always higher than $P_{\rm s}$ with CRA algorithm and SBA algorithm, $P_{\rm gs}$ with larger $k$ is always higher than that with smaller $k$.
	The difference among different $P_{\rm gs}$ corresponding to different $k$ is more significant with larger $M$.
	
	Fig. \ref{fig_14} shows the number of required time slots when $99\%$ of neighbors are discovered versus the number of mobile nodes $M$ with CRA algorithm, SBA algorithm, gossip-based algorithm and GSIM-ND algorithm when $p_{\rm t}$ is optimal, the number of modulation modes $k$ is $1$, $3$ and $5$ respectively.
	The algorithm convergence time increases as $M$ increases with CRA algorithm and SBA algorithm.
	This is due to the fact that the algorithm convergence time is increasing with the decrease of $P_{\rm s}$, the increase of $M$ and the number of neighbors $N$.
	The algorithm convergence time decreases firstly and then increases as the number of mobile nodes $M$ increases with gossip-based algorithm.
	This is due to the fact that $P_{\rm gs}$ increases firstly and then decreases with the increase of $M$.	
	The algorithm convergence time increases as $M$ increases with GSIM-ND algorithm.
	This is due to the fact that although $P_{\rm gs}$ increases firstly but then decreases with the increase of $M$.
	The decrement of algorithm convergence time due to $P_{\rm gs}$ is smaller than the increment due to $N$.
	Compared with gossip-based algorithm, GSIM-ND algorithm eliminates the influence of empty beam on ND when $M$ is small.
	Compared with CRA algorithm and SBA algorithm, the growth rate of convergence time of GSIM-ND algorithm is much smaller.
	
	\begin{figure*}[b]
		\setcounter{equation}{28}
		\hrulefill
		\begin{equation}\label{eq_29}
		\begin{aligned}
			{\overline{N_{\rm I}}} &= \int_{0}^{d}\int_{0}^{d}\int_{0}^{\sqrt{r^2-(s_2-s_1)^2}}N_{\rm I}(s_1,s_2,s_3)p_1(s_1)p_2(s_2)p_3(s_3)d(s_3)d(s_2)d(s_1)
			\\&=\frac{\rho}{d^2}\int_{0}^{d}\int_{0}^{d}\int_{0}^{\sqrt{r^2-(s_2-s_1)^2}}\int_{0}^{d}\frac{\left(\sqrt{r^{2}-(x-s_2)^{2}}		+\sqrt{r^{2}-(x-s_1)^{2}}	-s_3\right)}{\sqrt{r^2-(s_2-s_1)^2}}dxd(s_3)d(s_2)d(s_1)
		\end{aligned}
		\end{equation}
	\end{figure*}
	
	Regardless of low density or high density network, GSIM-ND algorithm can always maintain a short convergence time.
	Therefore, the stability of GSIM-ND algorithm is better than that of CRA algorithm, SBA algorithm and gossip-based algorithm.
	The stability of GSIM-ND algorithm with larger $k$ is better than that with smaller $k$,
	which is especially obvious when mobile nodes are densely distributed.
	
\section{Conclusion}\label{sec_7}
	This paper proposes GSIM-ND algorithm to satisfy the demands of multi-vehicle fast networking in VANET.	
	The expected number of discovered neighbors with a given period is derived.
	The expected bounds of the number of time slots when a given number of neighbors needs to be directly or indirectly discovered is derived as well.
	The simulation results verify the correctness of theoretical derivation.
	Moreover, we compare GSIM-ND algorithm with CRA algorithm, SBA algorithm and gossip-based algorithm.
	The efficiency and stability of GSIM-ND algorithm, gossip-based algorithm, CRA algorithm and SBA algorithm are discussed based on theoretical derivation and simulation.
	It is discovered that GSIM-ND is more efficient and stable than CRA algorithm, SBA algorithm and gossip-based algorithm.
	The convergence time of ND by GSIM-ND algorithm ($k=1$) is 90\% lower than that of CRA algorithm, 40\% lower than that of SBA algorithm, and 60\% lower than that of gossip-based algorithm when the density of vehicle nodes is low.
	The convergence time of ND by GSIM-ND algorithm ($k=1$) is 90\% lower than that of CRA algorithm, 60\% lower than that of SBA algorithm, and similar to that of gossip-based algorithm when the density of vehicle nodes is high.
	In addition, GSIM-ND algorithm can further reduces the convergence time of ND by adjusting $k$.
	Hence, GSIM-ND algorithm proposed in this paper can adapt to the high requirements of multi-vehicle fast networking better than the CRA algorithm, SBA algorithm and gossip-based algorithm in VANET.
	
	For the future work, we will take a further look at efficient long-term continuous ND to cope with the constant neighbor changes in VANET.

\begin{appendices}
\section{Number of Neighbors in a Non-empty Beam}
	In \cite{8}, the beam contains $q$ neighbors is defined as $q$-beam,
	the beam does not contain any neighbors is defined as empty beam.
	Since the location of mobile nodes can be obtained from the sensing information, most of the mobile nodes choose non-empty beams instead of empty beams with GSIM-ND algorithm.
	
	\cite{8} derived the expected number of neighbors in one beam based on the assumption that the neighbors are uniformly distributed.
	However, due to the restriction of road area in VANET, the neighbors of the mobile nodes in VANET are not uniformly distributed in each beam.
	
	\begin{lemma} \label{lemma4}
		When empty beams are not selected,
		the probability of a mobile node selecting any non-empty beam and the expected number of neighbors in a non-empty beam are respectively
		\setcounter{equation}{21}
		\begin{equation}\label{eq_22}
			\alpha = \frac{1}{B-E_{0}},
		\end{equation}
		\begin{equation}\label{eq_23}
			\overline{N_{\rm b}} = \sum\limits_{q=1}^{N}{\frac{q\cdot E_q}{B-E_0}} = \alpha \cdot \sum\limits_{q=1}^{N}{qE_q}.
		\end{equation}		
	\end{lemma}
	\begin{proof}
		The event that the number of $q$-beam is $e$ is called the event Q.		
		The number of situations where the nodes distribution satisfies the event Q is
		\begin{equation}\label{eq_24}
			SC_{e}^{q} = NS_{e}^{q}\cdot BS_{B-e}^{q},
		\end{equation}
		where $NS_{e}^{q}$ is the number of situations that $qe$ different nodes are selected from $N$ different nodes and put into the determined $e$ different beams, and these $e$ beams are $q$-beam.
		$BS_{B-e}^{q}$ is the number of situations where the remaining $N-qe$ different nodes are put into the selected $B-e$ different beams, and these $B-e$ beams are not $q$-beam.
		
		$NS_{e}^{q}$ and $BS_{B-e}^{q}$ are formulated as
		\begin{equation}\label{eq_25}
			NS_{e}^{q} = \binom{N}{qe}\cdot \prod_{w=0}^{e-1}\binom{q(e-w)}{q},
		\end{equation}
		\begin{equation}\label{eq_26}
			BS_{B-e}^{q} = \binom{B}{e}\cdot DS_{B-e}^{q}(N-qe),
		\end{equation}
		where $DS_{b}^{q}(n) = 0$ when $b=0$ or $n=0$.
		The expression of $DS_{b}^{q}(n)$ can be obtained in \cite{8} when $b\neq 0$ and $n\neq 0$.
		
		Then, the probability that the nodes distribution satisfies the event Q is
		\begin{equation}\label{eq_27}
			P(e,B,N,q) = \sum_{s=1}^{SC_{e}^{q}}P_{e}^{q}(s) = NS_{e}^{q}\cdot \sum_{s=1}^{BS_{B-e}^{q}}P_{e}^{q}(s),
		\end{equation}
		where $P_{e}^{q}(s)$ is the probability that the nodes distribution satisfies the $s$-th situation.
		
		The expected number of $q$-beam is defined as $E_q$.
		As discussed in \cite{8}, $E_q$ can be formulated through $P(e,B,N,q)$.
		Then, the probability of selecting any non-empty beam can be formulated in (\ref{eq_22}),
		the expected number of neighbors in a non-empty beam can be formulated in (\ref{eq_23}) when empty beams are not be selected in GSIM-ND Algorithm.
	\end{proof}
	\begin{figure*}[t]	
		\setcounter{equation}{34}
		\begin{equation}\label{eq_35}
			N_{\rm I}(s_1,s_2,s_3)=\rho \cdot \int_{0}^{d}\left(\sqrt{r^{2}-(x-s_2)^{2}}		+\sqrt{r^{2}-(x-s_1)^{2}}-s_3\right)dx
		\end{equation}
		\hrulefill
	\end{figure*}
	
\section{Number of Neighbors and Common Neighbors}	
	\begin{lemma} \label{lemma5}
		The average number of a mobile node's neighbors is
		\setcounter{equation}{27}
		\begin{equation}\label{eq_28}
		\begin{aligned}
			\overline{N} &= \int_{0}^{d}N(s_1)p_1(s_1)d(s_1)\\
			&=\frac{2\rho}{d}\int_{0}^{d}\int_{0}^{d}\sqrt{r^{2}-(x-s_1)^{2}}dxd(s_1).
		\end{aligned}
		\end{equation}
		
		The average number of common neighbors between two mobile nodes is formulated in (\ref{eq_29}).
	\end{lemma}
	
	\begin{proof}
		We need to assume a distribution of nodes when we analyze the performance of algorithm.
		Uniform distribution of nodes corresponds to the worst performance of ND algorithm since the empty beams exist with small probability.
		By analyzing and optimizing the algorithm based on Uniform distribution, we can enhance the robustness of the algorithm.
		Thus, we study the case that the mobile nodes are uniformly distributed.
		The algorithm proposed in this paper is valid for different distributions of node.
		
		The average density of mobile nodes in the network is
		\setcounter{equation}{29}
		\begin{equation}\label{eq_30}
			\rho = \frac{M}{Ld},
		\end{equation}
		where $M$ is the number of mobile nodes in the network, $L$ and $d$ are the road length and road width, respectively.
		In practice, $\rho$ can be sensed by the RSUs.
		
		\begin{figure}[h]
			\centering
			\includegraphics[width=0.5\textwidth]{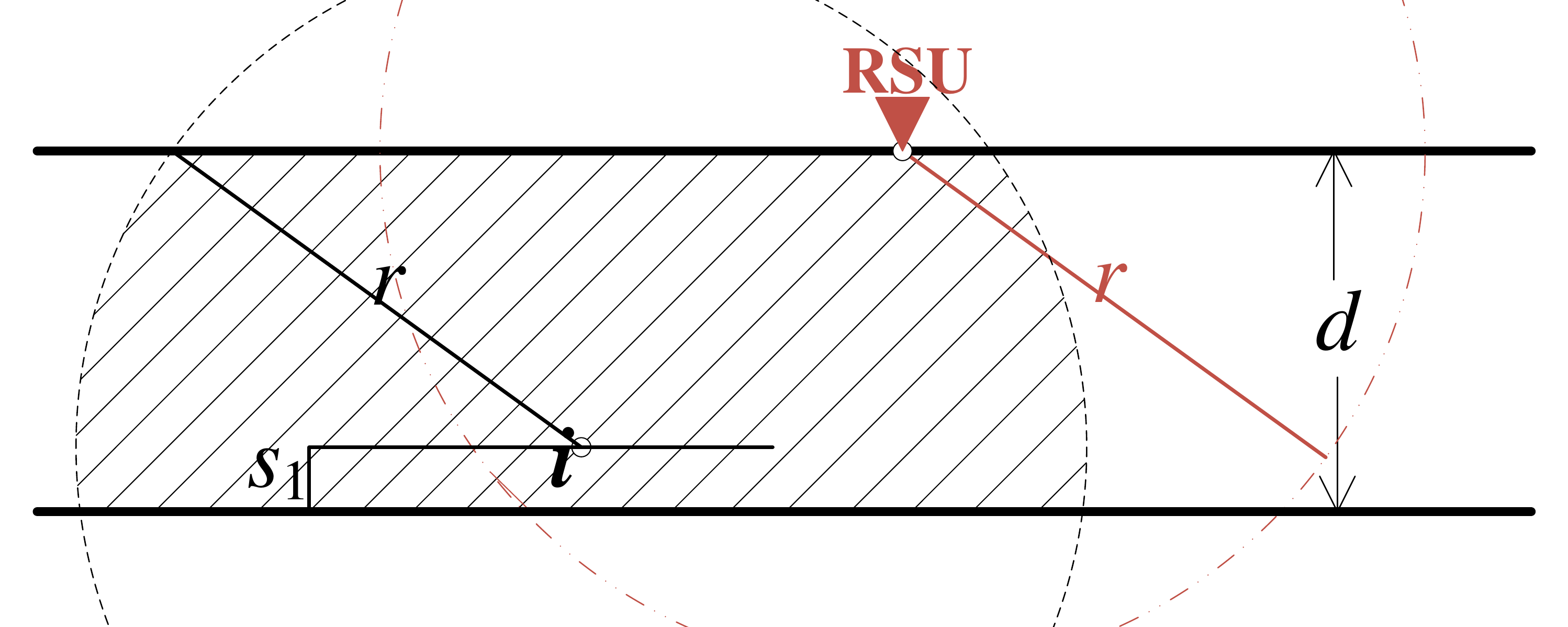}
			\caption{Fig. 15. Neighbor scope of RSU and mobile node.}
			\label{fig_15}
		\end{figure}
	
		As illustrated in Fig. \ref{fig_15}, the PDF of the distance from the mobile node \textit{i} to the edge of the road being $s_1$ is
		\begin{equation}\label{eq_31}
			{p_1(s_1)} = \frac{1}{d}.
		\end{equation}
		
		The number of neighbors of the mobile node \textit{i} in the shaded area shown in Fig. \ref{fig_15} can be approximated as
		\begin{equation}\label{eq_32}
			N(s_1)=\rho \cdot 2\int_{0}^{d}\sqrt{r^{2}-(x-s_1)^{2}}dx,
		\end{equation}
		where $r$ is the communication radius.
		
		The average number of neighbors $\overline{N}$ can be obtained by calculating the expectation of variable $N(s_1)$ when the value of $s_1$ is between the interval $[0,d]$.
		
		\begin{figure}[h]
			\centering
			\includegraphics[width=0.5\textwidth]{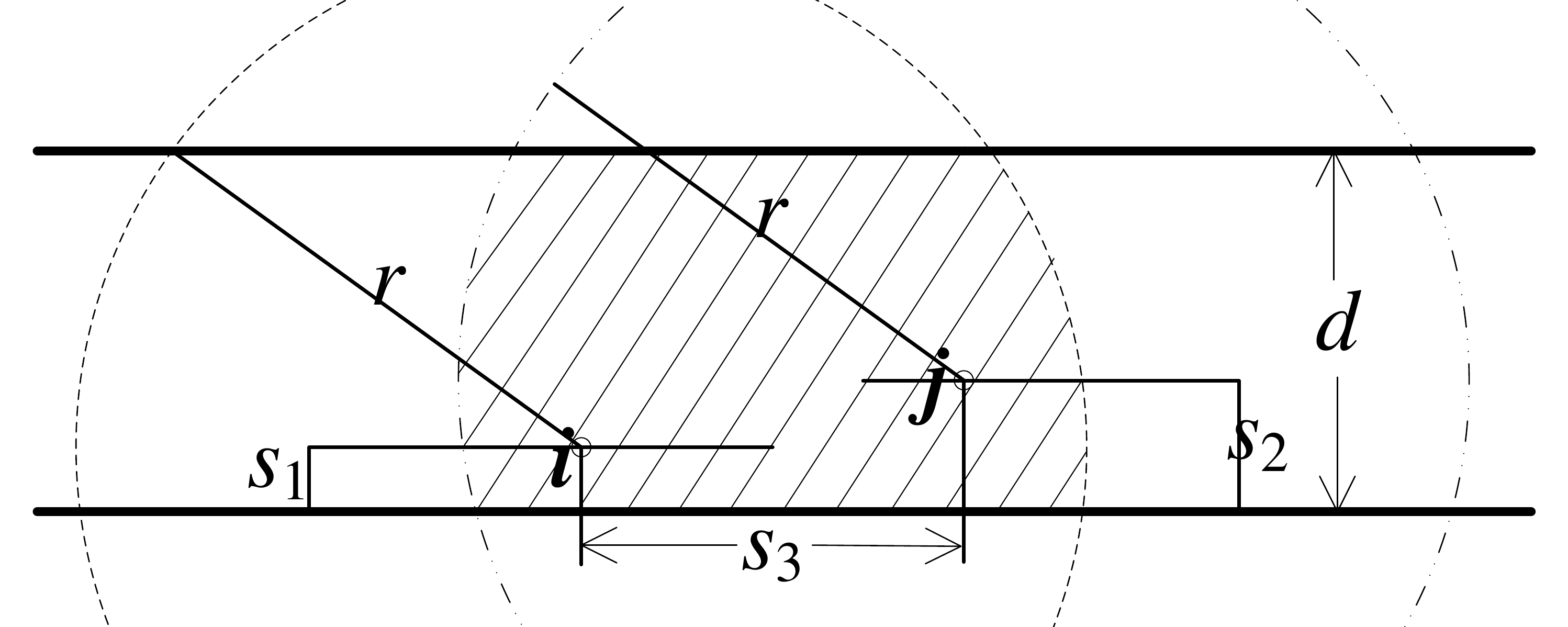}
			\caption{Fig. 16. The overlapped area between the neighborhoods of mobile node \textit{i} and \textit{j}.}
			\label{fig_16}
		\end{figure}
	
		As illustrated in Fig. \ref{fig_16}, the PDF of the distance from the mobile node \textit{j} to the edge of the road being $s_2$ is
		\begin{equation}\label{eq_33}
			{p_2(s_2)}={p_1(s_1)} = \frac{1}{d}.
		\end{equation}
		
		On the premise that the mobile node \textit{i} and mobile node \textit{j} are neighbors, the probability of the longitudinal distance component between the mobile node \textit{i} and \textit{j} being $s_3$ is
		\begin{equation}\label{eq_34}
			{p_3(s_3)} = \frac{1}{\sqrt{r^2-(s_2-s_1)^2}}.
		\end{equation}		
		
		The number of common neighbors between the mobile node \textit{i} and \textit{j} in the shaded area shown in Fig. \ref{fig_16} can be approximately formulated in (\ref{eq_35}).
		
		The average number of common neighbors $\overline{N_{\rm I}}$ can be obtained by calculating the expectation of the variable $N_{\rm I}(s_1,s_2,s_3)$ when the value of $s_1$ is between the interval $[0,d]$, the value of $s_2$ is between the interval $[0,d]$ and the value of $s_3$ is between the interval $[0,\sqrt{r^2-(s_2-s_1)^2}]$.
	\end{proof}
\end{appendices}


\begin{thebibliography}{00}
	\bibitem{1}
	Y. Ni, J. He, L. Cai and Y. Bo, ``Data Uploading in Hybrid V2V/V2I Vehicular Networks: Modeling and Cooperative Strategy," \textit{IEEE Transactions on Vehicular Technology}, vol. 67, no. 5, pp. 4602--4614, May. 2018.
	\bibitem{2}
	B. Ko, K. Liu, S. H. Son and K. Park, ``RSU-Assisted Adaptive Scheduling for Vehicle-to-Vehicle Data Sharing in Bidirectional Road Scenarios," \textit{IEEE Transactions on Intelligent Transportation Systems}, pp. 1--13, Jan. 2020.
	\bibitem{3}
	D. Kim, Y. Velasco, W. Wang, R. N. Uma, R. Hussain and S. Lee, ``A New Comprehensive RSU Installation Strategy for Cost-Efficient VANET Deployment," \textit{IEEE Transactions on Vehicular Technology}, vol. 66, no. 5, pp. 4200--4211, May. 2017.
	\bibitem{4}
	K. Dar, M. Bakhouya, J. Gaber, M. Wack and P. Lorenz, ``Wireless communication technologies for ITS applications [Topics in Automotive Networking]," \textit{IEEE Communications Magazine}, vol. 48, no. 5, pp. 156--162, May. 2010.
	\bibitem{5}
	J. Zhou, R. Q. Hu and Y. Qian, ``A Scalable Vehicular Network Architecture for Traffic Information Sharing," \textit{IEEE Journal on Selected Areas in Communications}, vol. 31, no. 9, pp. 85--93, Sept. 2013.
	\bibitem{6}
	T. Wang, L. Song and Z. Han, ``Coalitional Graph Games for Popular Content Distribution in Cognitive Radio VANETs," \textit{IEEE Transactions on Vehicular Technology}, vol. 62, no. 8, pp. 4010--4019, Oct. 2013.
	\bibitem{7}
	T. Wang, L. Song, Z. Han and B. Jiao, ``Dynamic Popular Content Distribution in Vehicular Networks using Coalition Formation Games," \textit{IEEE Journal on Selected Areas in Communications}, vol. 31, no. 9, pp. 538--547, Sept. 2013.
	\bibitem{8}
	Z. Wei, C. Han, C. Qiu, Z. Feng and H. Wu, ``Radar Assisted Fast Neighbor Discovery for Wireless Ad Hoc Networks," \textit{IEEE Access}, vol. 7, pp. 176514--176524, Oct. 2019.
	\bibitem{9}
	L. Galluccio, G. Morabito and S. Palazzo, ``Analytical evaluation of a tradeoff between energy efficiency and responsiveness of neighbor discovery in self-organizing ad hoc networks," \textit{IEEE Journal on Selected Areas in Communications}, vol. 22, no. 7, pp. 1167--1182, Sept. 2004.
	\bibitem{10}
	X. Duan, Y. Liu and X. Wang, ``SDN Enabled 5G-VANET: Adaptive Vehicle Clustering and Beamformed Transmission for Aggregated Traffic," \textit{IEEE Communications Magazine}, vol. 55, no. 7, pp. 120--127, Jul. 2017.
	\bibitem{11}
	3GPP TR 38.900 V15.0.0, ``Study on channel model for frequency spectrum above 6 GHz
	(Release 15)," June. 2018.
	\bibitem{12}
	3GPP TR 38.886 V16.3.0, ``User Equipment (UE) radio transmission and reception(Release 16)," Mar. 2021.
	\bibitem{13}
	J. B. Kenney, ``Dedicated Short-Range Communications (DSRC) Standards in the United States," in \textit{Proceedings of the IEEE}, vol. 99, no. 7, pp. 1162--1182, Jul. 2011.
	\bibitem{14}
	F. Liu, W. Yuan, C. Masouros and J. Yuan, ``Radar-Assisted Predictive Beamforming for Vehicular Links: Communication Served by Sensing," \textit{IEEE Transactions on Wireless Communications}, vol. 19, no. 11, pp. 7704--7719, Nov. 2020.
	\bibitem{15}
	A. Abdelaziz, C. Emre Koksal, R. Burton, F. Barickman, J. Martin, J. Weston and K. Woodruff, ``Beyond PKI: Enhanced Authentication in Vehicular Networks via MIMO," in \textit{Proc. 2018 IEEE 19th International Workshop on Signal Processing Advances in Wireless Communications (SPAWC)}, Kalamata, pp. 1--5, Jun. 2018.
	\bibitem{16}
	G. Lee and R. W. Heath, ``Fast Link Configuration for mmWave Multiuser MIMO Downlink Using Spatial AoD Angular Supports," in \textit{Proc. 2017 IEEE 86th Vehicular Technology Conference (VTC-Fall)}, Toronto, ON, pp. 1--2, Sept. 2017.
	\bibitem{17}
	A. Graff, A. Ali and N. González-Prelcic, ``Measuring radar and communication congruence at millimeter wave frequencies," in \textit{Proc. 2019 53rd Asilomar Conference on Signals, Systems, and Computers}, Pacific Grove, CA, USA, pp. 925--929, Nov. 2019.
	\bibitem{18}
	D. Vlastaras, T. Abbas, D. Leston and F. Tufvesson, ``Universal medium range radar and IEEE 802.11p modem solution for integrated traffic safety," in \textit{Proc. 2013 13th International Conference on ITS Telecommunications (ITST)}, Tampere, pp. 193--197, Nov. 2013.
	\bibitem{19}
	D. Burghal, A. S. Tehrani and A. F. Molisch, ``On Expected Neighbor Discovery Time With Prior Information: Modeling, Bounds and Optimization," \textit{IEEE Transactions on Wireless Communications}, vol. 17, no. 1, pp. 339--351, Jan. 2018.
	\bibitem{20}
	Z. Zhang and B. Li, ``Neighbor discovery in mobile ad hoc self-configuring networks with directional antennas: algorithms and comparisons," \textit{IEEE Transactions on Wireless Communications}, vol. 7, no. 5, pp. 1540--1549, May. 2008.
	\bibitem{21}
	Z. Wei, X. Liu, C. Han and Z. Feng, ``Neighbor Discovery for Unmanned Aerial Vehicle Networks," \textit{IEEE Access}, vol. 6, pp. 68288--68301, Sept. 2018.
	\bibitem{22}
	H. Cai, B. Liu, L. Gui and M. Wu, ``Neighbor discovery algorithms in wireless networks using directional antennas," in \textit{Proc. 2012 IEEE International conference on communications (ICC)}, pp. 767--772, Jun. 2012.
	\bibitem{23}
	L. Chen, Y. Li and A. V. Vasilakos, ``Oblivious neighbor discovery for wireless devices with directional antennas," in \textit{Proc. IEEE INFOCOM 2016 - The 35th Annual IEEE International Conference on Computer Communications}, San Francisco, CA, pp. 1--9, Apr. 2016.
	\bibitem{24}
	J. -. Park, S. -. Cho, M. Y. Sanadidi and M. Gerla, ``An analytical framework for neighbor discovery strategies in ad hoc networks with sectorized antennas," \textit{IEEE Communications Letters}, vol. 13, no. 11, pp. 832--834, Nov. 2009.
	\bibitem{25}
	S. Vasudevan, D. Towsley, D. Goeckel, and R. Khalili, ``Neighbor discovery in wireless networks and the coupon collector's problem," in \textit{Proc. ACM MobiCom}, pp. 181--192, Sept. 2009.
	\bibitem{26}
	N. Liu, L. Peng, R. Xu, J. Zhang, W. Zhao and J. Zhu, ``Neighbor Discovery in Wireless Network with Double-Face Phased Array Radar," in \textit{Proc. 2016 12th International Conference on Mobile Ad-Hoc and Sensor Networks (MSN)}, Hefei, 2016, pp. 434--439, Dec. 2016.
	\bibitem{27}
	J. Li, L. Peng, Y. Ye, R. Xu, W. Zhao and C. Tian, ``A Neighbor Discovery Algorithm in Network of Radar and Communication Integrated System," in \textit{Proc. 2014 IEEE 17th International Conference on Computational Science and Engineering}, Chengdu, pp. 1142--1149, Dec. 2014.
	\bibitem{28}
	G. Sun, F. Wu, X. Gao, G. Chen and W. Wang, ``Time-Efficient Protocols for Neighbor Discovery in Wireless Ad Hoc Networks," \textit{IEEE Transactions on Vehicular Technology}, vol. 62, no. 6, pp. 2780--2791, Jul. 2013.
	\bibitem{29}
	R. Khalili, D. L. Goeckel, D. Towsley and A. Swami, ``Neighbor Discovery with Reception Status Feedback to Transmitters," in \textit{Proc. 2010 Proceedings IEEE INFOCOM}, San Diego, CA, pp. 1--9, Mar. 2010.
	\bibitem{30}
	S. Vasudevan, J. Kurose and D. Towsley, ``On neighbor discovery in wireless networks with directional antennas," in \textit{Proc. IEEE 24th Annual Joint Conference of the IEEE Computer and Communications Societies.}, Miami, FL, vol. 4, pp. 2502--2512, Mar. 2005.
	\bibitem{31}
	A. Russell, S. Vasudevan, B. Wang, W. Zeng, X. Chen and W. Wei, ``Neighbor Discovery in Wireless Networks with Multipacket Reception," \textit{IEEE Transactions on Parallel and Distributed Systems}, vol. 26, no. 7, pp. 1984--1998, 1 Jul. 2015.
	\bibitem{32}
	C. Liu, G. Zhang, W. Guo and R. He, ``Kalman Prediction-Based Neighbor Discovery and Its Effect on Routing Protocol in Vehicular Ad Hoc Networks," \textit{IEEE Transactions on Intelligent Transportation Systems}, vol. 21, no. 1, pp. 159--169, Jan. 2020.
	\bibitem{33}
	A. Nahar, H. Sikarwar and D. Das, ``CSBR: A Cosine Similarity Based Selective Broadcast Routing Protocol for Vehicular Ad-Hoc Networks," in \textit{Proc. 2020 IFIP Networking Conference (Networking)}, Paris, France, pp. 404--412, Jun. 2020.
	\bibitem{34}
	H. P. de Moraes and B. Ducourthial, ``Adaptive inter-messages delay in vehicular networks," in \textit{2016 IEEE 12th International Conference on Wireless and Mobile Computing, Networking and Communications (WiMob)}, New York, NY, pp. 1--8, Oct. 2016.
	\bibitem{35}
	F. Lyu, H. Zhu, N. Cheng, H. Zhou, W. Xu, M. Li and X. Shen, ``Characterizing Urban Vehicle-to-Vehicle Communications for Reliable Safety Applications," \textit{IEEE Transactions on Intelligent Transportation Systems}, vol. 21, no. 6, pp. 2586--2602, June. 2020.
	\bibitem{36}
	F. A. Teixeira, V. F. e Silva, J. L. Leoni, D. F.Macedo and J. M.S.Nogueira, ``Vehicular networks using the IEEE 802.11p standard: An experimental analysis," \textit{Vehicular Communications}, pp. 91--96, Apr. 2014.
	\bibitem{37}
	F. Lyu, N. Cheng, H. Zhu, H. Zhou, W. Xu, M. Li and X. Shen, ``Towards Rear-End Collision Avoidance: Adaptive Beaconing for Connected Vehicles," \textit{IEEE Transactions on Intelligent Transportation Systems}, vol. 22, no. 2, pp. 1248--1263, Feb. 2021.
	\bibitem{38}
	F. Tian, B. Liu, H. Cai, H. Zhou and L. Gui, ``Practical Asynchronous Neighbor Discovery in Ad Hoc Networks With Directional Antennas," \textit{IEEE Transactions on Vehicular Technology}, vol. 65, no. 5, pp. 2414--2427, May. 2016.
	\bibitem{39}
	C. Baojian, Z. Dehai and X. Dazhi, ``An Improved Time Synchronous System Based on GPS Disciplined Rubidium," in \textit{Proc. 2010 International Conference on Intelligent Computation Technology and Automation}, Changsha, pp. 599--602, May. 2010.
	\bibitem{40}
	C. Chembe, D. Kunda, I. Ahmedy, R. M. Noor, A. Q. M. Sabri and M. A. Ngadi, ``Infrastructure based spectrum sensing scheme in VANET using reinforcement learning," \textit{Vehicular Communications 18}, pp. 100161.1--100161.14, May. 2019.
	\bibitem{41}
	A. Zhang, M. L. Rahman, X. Huang, Y. J. Guo, S. Chen and R. W. Heath, ``Perceptive Mobile Networks: Cellular Networks With Radio Vision via Joint Communication and Radar Sensing," \textit{IEEE Vehicular Technology Magazine}, vol. 16, no. 2, pp. 20--30, June. 2021.
	\bibitem{42}
	3GPP TR 23.703 V0.4.1, ``Study on Architecture Enhancements to Support Proximity Services (ProSe) (Release 12)," June. 2013.
	\bibitem{43}
	Z. Zhang, ``Performance of neighbor discovery algorithms in mobile ad hoc self-configuring networks with directional antennas," in \textit{Proc. MILCOM 2005 - 2005 IEEE Military Communications Conference}, Atlantic City, NJ, pp. 3162--3168, Oct. 2005.
	
\end{thebibliography}
\end{document}